\documentclass[11pt]{article}
\pdfoutput=1
\usepackage[margin=1in]{geometry}
\usepackage[utf8]{inputenc}
\usepackage{amsmath, amssymb, amsthm, dsfont}
\usepackage{enumitem}
\usepackage{graphicx, subfig}
\usepackage{hyperref}
\usepackage{comment}
\usepackage[dvipsnames]{xcolor}

\allowdisplaybreaks

\usepackage{color}

\title{Budget-Smoothed Analysis for Submodular Maximization\thanks{We thank Eric Balkanski and Matt Weinberg for interesting discussions and comments on earlier drafts.}}
\author{Aviad Rubinstein\thanks{Supported by NSF CCF-1954927, and a David and Lucile Packard Fellowship.}\\Stanford University\\\texttt{aviad@cs.stanford.edu} \and Junyao Zhao\thanks{Supported by NSF CCF-1954927.}\\ Stanford University\\\texttt{junyaoz@stanford.edu}}
\date{}

\usepackage{complexity}

\renewcommand \R {\mathbb{R}}
\renewcommand \E {\mathbb{E}}
\def \ie {{i.e.}}

\def \eps {\varepsilon}
\def \cD {\mathcal{D}}
\def \cP {\mathcal{P}}
\def \cR {\mathcal{R}}

\newtheorem{theorem}{Theorem}[section]
\newtheorem{lemma}[theorem]{Lemma}

\newtheorem{proposition}[theorem]{Proposition}

\theoremstyle{definition}
\newtheorem{definition}[theorem]{Definition}
\newtheorem*{remark*}{Remark}
\newtheorem{claim}[theorem]{Claim}

\newtheorem{example}[theorem]{Example}
\newtheorem{observation}[theorem]{Observation}

\def \fglb {f^{(\textrm{greedy-lb})}}

\begin{document}

\maketitle

\begin{abstract}
The greedy algorithm for monotone submodular function maximization subject to cardinality constraint is guaranteed to approximate the optimal solution to within a $1-1/e$ factor. Although it is well known that this guarantee is essentially tight in the worst case --- for greedy and in fact any efficient algorithm, experiments show that greedy performs better in practice. We observe that for many applications in practice, the empirical distribution of the budgets (i.e., cardinality constraints) is supported on a wide range, and moreover, all the existing hardness results in theory break under a large perturbation of the budget.

To understand the effect of the budget from both algorithmic and hardness perspectives, we introduce a new notion of {\em budget-smoothed analysis}. We prove that greedy is optimal for {\em every} budget distribution, and we give a characterization for the worst-case submodular functions. Based on these results, we show that on the algorithmic side, under realistic budget distributions, greedy and related algorithms enjoy provably better approximation guarantees, that hold even for worst-case functions, and on the hardness side, there exist hard functions that are fairly robust to all the budget distributions.
\end{abstract}

\setcounter{page}{0}
\thispagestyle{empty}
\newpage

\section{Introduction}

Monotone submodular function maximization subject to a cardinality constraint is a fundamental problem in combinatorial optimization with a wide variety of applications including feature selection, sensor placement, influence maximization in social networks, document summarization, etc. (see e.g.~\cite{KG14} and references therein). We will use influence maximization in social networks as a running example: an advertiser has a limited budget of $k$ free product samples that she wishes to distribute to seed consumers, who will then propagate the news about the product to their friends, then their friends' friends, etc. 
The standard approach to this problem~\cite{KKT15} models the expected final reach of the campaign as a monotone submodular function $f:\cP([n]) \rightarrow \mathbb{R}$ of the set of seed consumers (where $[n]$ is the set of all users in the network).
The goal of the optimization problem is to find a set of $k$ seed consumers that (approximately) maximizes $f$. 

Classic work shows that the simple greedy algorithm achieves a $1-1/e$-approximaiton to the optimal solution in the worst case~\cite{NemhauserWF78}. Furthermore, this bound is tight for algorithms that make sub-exponential queries to the function~\cite{NemhauserW78,Vondrak13}; 
and even succinctly representable functions (e.g. simple models of influence propagation on a social network graph) do not allow better approximation algorithms unless $\P=\NP$~\cite{Feige98}.
In theory, this tight characterization of the optimal approximation factor is very satisfying.

Given the importance of this problem in practical applications, it is also interesting to ask what is the optimal approximation factor that can be obtained on realistic instances. As one can expect, the performance of the greedy algorithm tends to be significantly better in practice (e.g.~\cite{TSP20, BQS21}). 
When reasoning about real-world instances, there is a natural tradeoff between quality and generality of the guarantees: at one extreme, worst-case analysis only gives a $(1-1/e)$-approximation but applies to every instance; at the other extreme we could, in principle%
\footnote{In general calculating the approximation factor on a real-world instance requires computing the value of the optimal solution. The recent work~\cite{BQS21} provides an instance-specific method to estimate the optimal value.}, empirically evaluate the performance of the greedy algorithm on each instance of interest, but we would have to redo this for every new instance.
Ideally, we want to extend the classic worst-case model --while making minimal assumptions-- to explain why efficient algorithms like greedy should obtain better-than-$(1-1/e)$-approximation in practice. 

Coming up with useful and realistic assumptions about submodular functions continues to be an interesting and active topic of research. In Section~\ref{sec:related} we survey several natural restrictions, including recent success stories that allow for improved approximation algorithms~\cite{KL14,SviridenkoVW17,BRS16,HorelS16,Yoshida16,CRV17,TSP20,STY20}. Deferring details for later, we argue that the bottom line of this discussion is that submodular functions are complex  objects, and as such modeling their beyond-worst-case behavior is tricky and application-dependent, and moreover, it is often intractable to verify the model assumptions in practice.

In this work, we consider beyond-worst-case analysis (and hardness) of submodular maxmization from a novel perspective by focusing on modeling the average-case behavior of a much simpler object: the cardinality constraint.
As we now explain, our approach of perturbing the cardinality constraint is motivated by both theory and practice.

In practice, many applications of submodular maximization have multiple users with the same or similar objective but various budgets (i.e., cardinality constraints). In the example about influence maximization, multiple advertisers could advertise and propagate on the same social network and hence maximize essentially the same, possibly worst-case, submodular function. However, their budgets can easily vary by an order of magnitude or more because of different sizes of business or different amounts of funds. A concrete example is the distribution of the campaign budgets of the candidates in the 2020 Democratic Party primary elections~\cite{2020election}. Thus even if the social network/submodular function is worst-case, the ``average'' advertiser uses an ``average'' budget which is independent of the social network/submodular function. In a different example about feature selection, the engineers wish to make predictions in the testing phase using a small subset from the high-dimensional feature space that is selected during the training phase, and in the training phase, they apply the same standard machine learning model (e.g., linear regression, logistic regression) to the same standard datasets (e.g., ImageNet~\cite{ImageNet}) and hence optimize the same monotone (approximately) submodular objective~\cite{DK11,EKDN18}. However, the number of features they want to choose can easily range from one hundred to one million depending on the computational power they have or the model complexity they prefer. Therefore, from a practical point of view, it is interesting to understand whether the average-case behavior of the budget makes the problem of submodular maximization easier to some extent.

In theory, all the known worst-case instances for cardinality-constrained monotone submodular maximization~\cite{NemhauserWF78,Vondrak13,Feige98} are sensitive to large budget perturbations: even outputting a random solution achieves $\gtrsim0.95$ approximation, when we perturb (i.e., multiply) the cardinality constraint by a significant multiplicative factor\footnote{Clearly, tiny perturbations of the constraint cannot escape the hardness of approximation results, because by submodularity, a $(1+\eps)$-multiplicative perturbation in budget cannot affect the value of the solution by more than a $(1+\eps)$-factor.} like $0.1$ or $10$.
Hence, from a theoretical point of view, it is interesting to investigate the effect perturbing the cardinality constraint has on the hardness of approximation results.

With the above motivations from theory and practice, we initiate the study of submodular function maximization in a semi-adversarial setting, where the (empirical) distribution of the cardinality constraints is supported on a wide range (e.g. $[x,10x]$), from both the algorithmic and the hardness perspectives. Namely, we hope to answer the question of whether a large random perturbation of the cardinality constraint allows efficient algorithms to achieve higher optimal approximation ratio or there is a stronger hardness result that is robust to any such perturbations.

To formalize this question, we propose a simple and elegant framework called {\em budget-smoothed analysis}. The name is inspired by the celebrated smoothed analysis for linear programming (LP)~\cite{SpielmanT04}.
Admittedly, the analogy is not perfect: we consider much larger perturbations than smoothed analysis for LP.
However, much like smoothed analysis for LP, the generative process of random perturbations for budget-smoothed analysis for submodular maximization is grounded in concrete applications (such as viral marketing and feature selection -- see discussion above).

\subsection*{The budget-smoothed analysis model}
We study monotone submodular function maximization subject to a cardinality constraint in the following semi-adversarial setting:
\begin{definition}[Budget-smoothed analysis] \hfill

\begin{enumerate}
\item The distribution $\tilde{\cD}$ of budgets (e.g. uniform over $[x,10x]$) is given as input to the adversary.
\item The adversary chooses a (monotone submodular) function $f$.
\item The cardinality constraint $k \sim \tilde{\cD}$ is drawn at random and given as input to the algorithm.
\item The algorithm (approximately) maximizes $f$ over all sets of size at most $k$.
\end{enumerate}
For any distribution $\tilde{\cD}$, we're interested in the expected ratio $\cR_{ALG}(f,\tilde{\cD})$ between the value obtained by the algorithm and the optimal solution,
$$ \cR_{ALG}(f,\tilde{\cD}):=\E_{k \sim \tilde{\cD}}\left[\frac{f(ALG)}{f(OPT)}\right].$$
\label{def:budget_smoothed}
\end{definition}



For notational convenience, we make the following change to the above model: Rather than sampling the budget from a distribution, each instance will be characterized by a base budget $k_0$ and a {\em budget perturbation distribution} $\cD$, with the final cardinality constraint being $k:= \rho \cdot k_0$ for $\rho \sim \cD$. This will allow us to talk about a distribution $\cD$ like ``uniform over $[x,10x]$'' while studying the asymptotic complexity as the instance size and cardinality constraint go to infinity.

The budget-smoothed analysis model has the following advantages --- First, it is simple and clean in theory, which not only provides a formal setup for studying our aforementioned question but also has the flexibility to be integrated with other models that make beyond-worst-case assumptions about the submodular functions. Second, it is easy-to-apply in practice, since calculating the (empirical) distribution of the budgets is much more tractable than verifying the complex assumptions about the submodular functions. 

\subsection*{Our results}
We revisit the classic problem of monotone submodular maximization in our model of budget-smoothed analysis, and in particular, we investigate the following fundamental questions:
\begin{description}
\item[Question 1:] What is the optimal efficient algorithm?
\item[Question 2:] What are the worst-case instances for an arbitrary budget distribution?
\item[Question 3:] What are the optimal approximation factors for the budget distributions that are supported on a wide range (e.g. $[x,10x]$) and what is the best budget distribution?
\end{description}

\begin{remark*}
All the hardness results below hold both in the black-box oracle model (for any algorithm that makes a subexponential number of queries), or assuming $\P\neq\NP$ in the computational model for coverage functions (on a polynomial-size graph).
\end{remark*} 

\paragraph{Result 1 (main theorem): Optimal approximation algorithms}
Our main theorem shows that a large class of algorithms that are (near-)optimal in the classic setting continue to obtain (near-)optimal approximation factors under budget-smoothed analysis for {\em any} distribution (Theorem~\ref{thm:greedy_is_optimal} and Observation~\ref{obs:opt_algorithms}).
This class includes the classic greedy algorithms, as well as (variants of) recent efficient parallel algorithms, and Map-Reduce algorithms (Appendix~\ref{sec:practical_algorithms}). In particular, these algorithms are optimal even in comparison to algorithms that know the budget perturbation distribution $\cD$. In other words, they intrinsically adjust themselves to the budget distribution optimally. 
The proof of the main theorem relies on a characterization of the worst-case instances in the model of budget-smoothed analysis, and therefore, we completely answer Questions 1 and 2.

The main theorem and Question 3 set up a win-win scenario for us: either we explain (a fraction of) the success of greedy for some interesting distributions, or we get a stronger hardness result that hold against all the budget distributions. Either way, we would bring new insights to the classic problem of submodular maximization.

Applying the main theorem, we manage to give partial answers (Result 2 and 3) to Question 3 on both positive and negative sides.

\paragraph{Result 2: Optimal approximation factors}
For any budget perturbation distribution $\cD$, we formulate a simple (but non-convex) mathematical program (Section~\ref{sec:simulation}) that computes the optimal possible approximation factor for a given budget distribution. We also include some numerical estimates for natural distributions (Table~\ref{table:simulation_result}). For the special case of $\cD$ supported on two budgets, we also give a closed-form solution (Proposition~\ref{prop:two_budget}).

These results are interesting on both positive and negative sides --- On the positive side, the optimal approximation ratios have modest but non-negligible improvements for many interesting distributions {\em even for worst-case submodular functions}, which explain a fraction of the success of greedy algorithms. On the negative side, we pin down the worst-case instances for these distributions which remain significantly hard to approximate (and studying these instances might provide new insights about the structure of beyond-worst-case submodular functions).

\paragraph{Result 3: Bounding the best-case budget distribution}
We also prove that for every budget distribution and any efficient algorithm, the optimal budget-smoothed analysis approximation
factor is bounded away from 1, and in particular, it is at most 0.9087 (Theorem~\ref{thm:lower_bound}). 

It is worth mentioning that because the program in Result 2 is non-convex, we are only able to compute the optimal approximation factors after discretizing the budget distributions with a limited number of budgets, and thus, the positive results may still have a lot of room to improve\footnote{For comparison,~\cite{SpielmanT04}'s original polynomial upper bound for smoothed analysis (of a non-trivial variant) of the Simplex algorithm was $\tilde{O}(d^{55}n^{86}\sigma^{-30} + d^{70}n^{86})$ iterations ($d$ is the number of variables, $n$ is the number of constraints, and $\sigma^2$ is the variance of the perturbation), which is significantly improved now~\cite{DH18}. That said, in light of Result 3, we do not expect the positive side of budget-smoothed analysis to fully explain the success of greedy in practice, but explaining a greater fraction of success or showing a more robust hardness result would still be interesting.}. This leaves an interesting open problem: close the gap and give a complete answer to Question 3.

\begin{table}
\begin{center}
\caption{\label{table:simulation_result} Empirical Results}
\begin{tabular}{ |c|c| } 
 \hline
{\bf  Budget perturbation distribution} & \textbf{Worst-case approximation ratio} \\

 \hline
 Baseline (no perturbation) &  0.6321\\
 \hline
  \hline
 Uniform over $[1,10]$ & 0.6675 \\
 Log-scale-uniform over $[1,10]$ & 0.6674 \\ 
 Log-scale-uniform over $[1,600]$ & 0.6808 \\ 
 \hline
Top 10 social/political campaigns on Facebook & 0.6625 \\
 2020 Democratic presidential candidates & 0.6727 \\
 \hline
\end{tabular}
\end{center}
We calculated tight approximation factors for worst-case monotone submodular functions for several exemplary budget distributions. See Section~\ref{sub:numerical} for details. 
\end{table}

\subsection{Broader discussion}
Our model of budget-smoothed analysis introduces a new (and more tractable) angle for studying beyond-worst-case analysis and average-case hardness. We believe that there are countless future directions and applications to explore in the broader field of TCS.
To exhibit this breadth of possibilities, we mention a couple of preliminary results that we have for other problems that fit into our new model:
\begin{itemize}
\item Submodular maximization subject to {\em knapsack constraint}. While the optimal $1-1/e$ factor can again be recovered in polynomial time, the state-of-the-art algorithms for this problem are still not completely satisfying~\cite{Sviridenko04,EN19,NS20}, and the greedy algorithm does not provide any non-trivial approximation guarantee. Our preliminary results show that with budget-smoothed analysis, greedy guarantees a constant factor approximation with knapsack constraint, and in fact to date we haven't been able to rule out $1-1/e$ approximation (or better).
\item Budget-feasible mechanism design: this is a well-studied problem in algorithmic game theory~\cite{singer2010budget,chen2011approximability,DobzinskiPS11,BadanidiyuruKS12,SingerM13,ChanC14,EG14,GoelNS14,HorelIM14,BalkanskiH16,ChanC16,NushiS0K16,ZhaoLM16,ZhengWGZTC17,LeonardiMSZ17,anari2018budget,KhalilabadiT18,AmanatidisKS19,gravin2019optimal,LiZY20}. Under a large market assumption~\cite{anari2018budget} obtain a mechanism with optimal approximation guarantee, incidentally also $1-1/e$. Our preliminary results show that this mechanism does {\em not} improve at all under budget-smoothed analysis. However, the budget-smoothed analysis inspires a {\em new mechanism} that is not only optimal for every budget distribution but also instance-optimal among a canonical class of mechanisms, and it also significantly outperforms~\cite{anari2018budget}'s mechanism on realistic distributions empirically. 
\end{itemize}
In hindsight, although the performance improvement guaranteed by budget-smoothed analysis is relatively moderate, we believe that the optimality of an algorithm in the model of budget-smoothed analysis is a theoretical evidence that the algorithm is not just worst-case optimal but also likely to perform favorably on realistic instances. In other words, budget-smoothed analysis offers an analytically approachable beyond-worst-case performance test for the worst-case optimal algorithms of budget-constrained problems, which helps us identify (or design) the ``right'' algorithm among various worst-case optimal algorithms.

\subsection{Roadmap}
In Section~\ref{sec:related} we survey several other approaches to beyond-worst-case submodular maximization. In Section~\ref{sec:greedy} we prove our core technical result
, namely that greedy obtains optimal approximation factors for any distribution; in Appendix~\ref{sec:practical_algorithms} we extend this result to other related algorithms. Henceforth, we build on these techniques; in particular we simply analyze the approximation factors of the greedy algorithm. In Section~\ref{sec:bound}, we prove 
that the optimal approximation factor is bounded away from $1$ for any distribution. In Section~\ref{sec:simulation}, we 
characterize the optimal approximation factor by a program, which we then use to simulate several exemplary distributions. 

\subsection{Beyond-worst-case submodular functions}\label{sec:related}
Due to the popularity of submodular maximization in practice, there is a lot of interest in understanding and designing algorithms for ``typical'' cases.
We discuss a few approaches below. 
We note that our model of beyond-worst-case {\em cardinality constraint} is orthogonal to any assumptions about the submodular function, and in principle could be combined with any of them to obtain even stronger results.

The model most closely in spirit to our smoothed-analysis-like approach is to take a worst-case submodular function and perturb it with random noise. 
The most straightforward way of doing this is independently perturbing the value of the function for each set. Unfortunately, this breaks the submodularity, which makes the problem significantly {\em harder}, even for small perturbations:~\cite{HS17} barely recovers the $1-1/e$ approximation factor in this setting (under further restrictions and with a technically involved algorithm).

Another approach is to consider {\em coverage functions}, an important special class of monotone submodular functions. This restriction has been successful for learning submodular functions~\cite{BalcanCIW12, BadanidiyuruDFKNR12, FeldmanK14}, but Feige's \NP-hard instance already rules out efficient algorithms with improved approximation ratios for this case.
One may combine this restriction with perturbations of the weights of the elements of the ground set; but it is not hard to show that Feige's instance can be made robust even to very large amounts of noise.
Another alternative is to consider special classes of graphs, e.g., power law, small-world, or triangle-dense that are common for social networks ~\cite{WS98,GRS16}. But again Feige's instance either already satisfies all of those, or can be adapted to do so.

Another popular restriction of monotone submodular functions is {\em bounded curvature}~\cite{CC84}, which restricts the extent to which ground elements interact; this indeed allows for better algorithms with applications to e.g.~maximum entropy sampling~\cite{CC84,SviridenkoVW17,BRS16,HorelS16,Yoshida16}. But bounded curvature seems too restrictive for applications like influence in social networks and consumers' valuations with diminishing returns%
\footnote{If, for example, we already selected all of a node's neighbors, the  marginal contribution of adding this node is diminished to zero. For consumers' valuations, the marginal contribution of, e.g. the one-thousandth apple, is again diminished to essentially zero. This means that curvature is unbounded in both settings (see~\cite{CC84,SviridenkoVW17} for formal definitions).}.

To cope with the limited applicability of curvature, the original paper of \cite{CC84} also defined a relaxed notion of {\em greedy curvature}, which only restricts the interaction between elements selected by the greedy algorithm and elements in the optimal solution. In exciting recent work,~\cite{TSP20} define various notions of {\em sharpness} which only restricts the interactions of the average element of the optimal solution. Both greedy curvature and sharpness parameters suffer from the disadvantage that they may be intractable to compute (both are assumptions about interaction of elements with the optimal solutions, and if we knew the optimal solution...). Moreover, due to their complicated form it's hard to heuristically reason about their fit for any particular application. Nevertheless, on the positive side both are more realistic than vanilla curvature assumption, and combining them with our budget-smoothed analysis model is an interesting direction for future research.

\cite{CRV17} study submodular maximization under a {\em stability} assumption, i.e.~they assume that the optimal subset does not change when the function is perturbed. 
\cite{TSP20} argue that in the context of submodular maximization, stable instances may fail to capture significant interaction between elements. As in the case of greedy curvature and sharpness, it is also not clear how to compute the stability of a function, or reason about instances that we expect to be stable.

Finally, one setting that is both natural and allows for improved approximation factors is influence maximization in undirected graphs~\cite{KL14, ST19, STY20}. Specifically,~\cite{KL14} prove that the greedy algorithm obtains a $(1-1/e+\varepsilon)$-approximation (for some small unspecified constant $\varepsilon>0$) for the independent cascade model on undirected graph. \cite{STY20} show that in the linear threshold model the greedy algorithm does not beat the $(1-1/e)$-approximation factor (by any constant, in the worst case). 

\section{Preliminaries}
\begin{definition}
A function $f:2^{V}\to\R_{\ge0}$ is submodular if for all $S\subseteq T\subseteq V$ and $i\in V\setminus T$, $f(S\cup \{i\})-f(S)\ge f(T\cup \{i\})-f(T)$, where $V$ is called ground set. Moreover, we denote the marginal gain by $f(X\mid S):=f(X\cup S)-f(S)$.
\end{definition}
We make the following conventions in this paper---When we say \textbf{``efficient algorithm''}, we mean polynomial time algorithms in the general computation model assuming $\P\neq\NP$, or algorithms using sub-exponential number of function queries in the oracle query model. Moreover, we consider \textbf{continuous distribution of budget perturbations} $\mathcal{D}$, and we let $\mathcal{D}(k)$ denote the distribution of budgets in which a budget is sampled by multiplying a random perturbation factor $\rho\sim \mathcal{D}$ with $k$ (if $\rho\cdot k$ is fractional, we can round it to an integer). Furthermore, following Definition~\ref{def:budget_smoothed}, we denote 
$
    \cR_{ALG}(\cD(k)):=\min_{f}\cR_{ALG}(f,\cD(k))\,\,\textrm{s.t. $f$ is monotone and submodular}.
$

The hard instances in our analysis can be built on top of either Feige's max-$k$-cover instances~\cite{Feige98} or Vondr\'ak's hard instances~\cite{Vondrak13}. In the following theorem, we summarize useful properties of these two hardness results.
\begin{theorem}\label{thm:basic_hardness}
There exists a class of monotone submodular functions $\mathcal{C}$ such that for every $\epsilon>0$, for any efficient algorithm $\mathcal{A}$ that given a submodular function $f$ and an integer $l$ outputs a set $X_l$ of cardinality $l$, for every sufficiently large $k$ that grows with size of instance, there is a submodular function $f_k\in\mathcal{C}$ such that
\begin{enumerate}[label=(\roman*)]
    \item for all $l\le k$, $f_k(O_l)=(l/k)f_k(O_k)$, where $O_l$ is the optimal set that maximizes $f_k$ among all cardinality-$l$ sets, and
    \item for all $l$, $f_k(X_l)\le (1-e^{-l/k}+\epsilon)f_k(O_k)$.
\end{enumerate}
\end{theorem}
Next, we state a standard lemma for greedy analysis.
\begin{lemma}\label{lem:greedy_one_step}
Given a monotone submodular $f$, we let $X_k$ and $O_k$ denote the greedy solution and the optimal solution of cardinality $k$, respectively. Then, for all $i,k> 0$, $f(X_i)-f(X_{i-1})\ge \frac{1}{k}(f(O_{k})-f(X_{i-1}))$.
\end{lemma}
\begin{proof}
Let $x_i$ denote the $i$-th element selected by greedy. It holds that
\begin{align*}
    f(X_i)-f(X_{i-1})&=f(x_i\mid X_{i-1}) \ge \frac{1}{k}\cdot\sum_{o\in O_{k}} f(o\mid X_{i-1})\\
    &\ge \frac{1}{k}\cdot f(O_{k}\mid X_{i-1})\ge \frac{1}{k}(f(O_{k})-f(X_{i-1})),
\end{align*}
where the first inequality is by greedy selection, the second is by submodularity, and the third is by monotonicity.
\end{proof}

\section{Greedy is Optimal}\label{sec:greedy}

In this section, we prove our core technical result: greedy is optimal for submodular maximization with respect to arbitrary distribution of budget perturbations.
\begin{theorem}\label{thm:greedy_is_optimal}
For any distribution of budget perturbations $\mathcal{D}$, for every $\epsilon'>0$, for any efficient algorithm $\mathcal{A}$, for every sufficiently large\footnote{The implicit dependence of $k$ on $\epsilon$ in Theorem~\ref{thm:basic_hardness} carries over to the dependence of $k$ on $\epsilon'$ in this statement, and therefore, we keep such dependence implicit, and we are mostly interested in the asymptotic result.} $k$ that grows with the size of instance, it holds that $\cR_{\mathcal{A}}(\mathcal{D}(k))\le(1+\epsilon')\cR_{\textrm{greedy}}(\mathcal{D}(k))$.
\end{theorem}

Theorem~\ref{thm:greedy_is_optimal} follows directly from Theorem~\ref{thm:many_budgets} using a discretization argument. 

\begin{theorem}\label{thm:many_budgets}
For any perturbation factors $0<\rho_1<\rho_2<\dots<\rho_{m}$, for every $\epsilon>0$, there exists a sufficiently large $k$ that grows with the size of instance such that given $m$ budgets $k_1=\rho_1\cdot k$, \dots, $k_{m}=\rho_{m}\cdot k$, for any monotone submodular function $f^{\textrm{(bad-for-greedy)}}$, for any efficient algorithm $\mathcal{A}$, there exists a monotone submodular function $f$ such that
\begin{description}
    \item[(i) Greedy is (almost) no worse than $\mathcal{A}$ on $f$:] for all $i\in[m]$, the solution $Y_{k_i}$ computed by $\mathcal{A}$ for budget $k_i$ has value $f(Y_{k_i})\le (1+\epsilon)f(X_{k_i})$, where $X_{k_i}$ is the greedy solution for budget $k_i$,
    \item[(ii) $f$ is as hard as $f^{\textrm{(bad-for-greedy)}}$ for greedy:] for all $i\in[m]$, given budget $k_i$, the approximation ratio of greedy on $f$ is at most $1+\epsilon$ times the approximation ratio of greedy on $f^{\textrm{(bad-for-greedy)}}$.
\end{description}
\end{theorem}

\begin{proof}[Proof of Theorem~\ref{thm:greedy_is_optimal}]
For arbitrarily small $\tau>0$, let $\rho_{\min}$ and $\rho_{\max}$ be such that the mass of $\mathcal{D}$ on $[\rho_{\min}, \rho_{\max}]$ is at least $1-\tau$.  We discretize $\{\rho_{\min}\cdot k,\,\rho_{\min}\cdot k+1,\,\dots,\rho_{\max}\cdot k\}$ into $\{\rho_{\min}\cdot k,\,(1+\delta)\rho_{\min}\cdot k,\,(1+\delta)^2\rho_{\min}\cdot k,\,\dots,\,\rho_{\max}\cdot k\}$. Without loss of generality, we assume that there exists $m$ such that $(1+\delta)^{m-1}\rho_{\min}=\rho_{\max}$ and every $k_i:=(1+\delta)^{i-1}\rho_{\min}\cdot k$ is integral. Let $f^{*}$ be the worst-case monotone submodular function for which greedy achieves only $\cR_{\textrm{greedy}}(\mathcal{D}(k))$ approximation in expectation.
By Theorem~\ref{thm:many_budgets}, for any efficient algorithm $\mathcal{A}$, there is a monotone submodular function $f$ such that for all $i\in[m]$, the solution $Y_{k_i}$ outputted by $\mathcal{A}$ for budget $k_i$ only achieves $f(Y_{k_i})\le (1+\epsilon)f(X_{k_i})$, where $X_{k_i}$ is the greedy solution for budget $k_i$, and moreover, for every budget $k_i$,
\begin{equation}\label{eq:no_worse_ratio}
    \frac{f(X_{k_i})}{f(O_{k_i})}\le (1+\epsilon)\frac{f^{\textrm{*}}(X^{*}_{k_i})}{f^{*}(O^{*}_{k_i})}
\end{equation}
where $O_{k_i}$ and $O^{*}_{k_i}$ denote optimal size-$k_i$ sets of $f$ and $f^{*}$ respectively, and $X^{*}_{k_i}$ denotes the size-$k_i$ greedy solution for $f^{*}$.

Besides, because marginal gain in each iteration of greedy is non-increasing, we have that $(1+\delta)f(X_{k_{i-1}})\ge f(X_{k_i})$. Furthermore, without loss of generality, we assume that $f(Y_{b})$ is non-decreasing in $b$, since otherwise, for budget $b$, we can let the algorithm choose the best solution among $Y_l$ for all $l\le b$ instead. For any $2\le i\le m$ and any budget $b$ such that $k_{i-1}\le b\le k_i$, it follows that
\begin{align*}
\frac{f(Y_{b})}{f(O_b)} &\le \frac{f(Y_{k_i})}{f(O_b)} && \text{(Since $b\le k_i$)}\\
&\le (1+\epsilon)\frac{f(X_{k_i})}{f(O_b)}&& \text{(By Theorem~\ref{thm:many_budgets})}\\
&\le (1+\epsilon)(1+\delta)\frac{f(X_{k_{i-1}})}{f(O_b)}&& \text{(Non-increasing marginal gain)}\\
&\le (1+\epsilon)(1+\delta)\frac{f(X_{k_{i-1}})}{f(O_{k_{i-1}})}&& \text{(Since $b\ge k_{i-1}$)}\\
&\le (1+\epsilon)^2(1+\delta)\frac{f^{*}(X^{*}_{k_{i-1}})}{f^{*}(O^{*}_{k_{i-1}})}&& \text{(By Eq.~\eqref{eq:no_worse_ratio})}\\
&\le (1+\epsilon)^2(1+\delta)\frac{f^{*}(X^{*}_{b})}{f^{*}(O^{*}_{k_{i-1}})}&& \text{(Since $X^{*}_{k_{i-1}}\subseteq X^{*}_{b}$)}\\
&\le (1+\epsilon)^2(1+\delta)^2\frac{f^{*}(X^{*}_{b})}{f^{*}(O^{*}_{b})}.&& \text{($f^{*}(O^{*}_{b})\le \frac{b}{k_{i-1}}f^{*}(O^{*}_{k_{i-1}})$ by submodularity)}
\end{align*} 

Therefore, for every budget $b$ in $\{\rho_{\min}\cdot k,\,\rho_{\min}\cdot k+1,\,\dots,\rho_{\max}\cdot k\}$, 
$\mathcal{A}$ can achieve on $f$ in expectation at most a factor of $(1+\epsilon)^2(1+\delta)^2$ times what greedy achieves on $f^{*}$.
The proof finishes since $\delta, \epsilon, \tau$ can be arbitrarily small.
\end{proof}

\subsection{Proof of Theorem~\ref{thm:many_budgets}}

In our proof we will not derive the analytic formula of the approximation ratio, but instead, the proof works in a black-box way---First, we introduce an array of parameters such that every instance can be characterized by these parameters, and we can show a parameterized guarantee of the marginal gain for each iteration of greedy. Then, we construct a hard instance characterized by the same parameters such that the best possible marginal gains for this instance always match the parametrized guarantees from greedy. It follows that the performance of greedy is optimal for every budget. 
Our hard instance has the following nice structure: it is a convex combination of disjoint-support copies of the classic hard instances guaranteed by Theorem~\ref{thm:basic_hardness}.

\begin{proof}[Proof of Theorem~\ref{thm:many_budgets}]
\subsubsection*{Proof setup: bounding a single step of greedy performance}
We first lower bound the single-step performance of greedy solutions. 
By Lemma~\ref{lem:greedy_one_step}, we have the following performance guarantees for each iteration of greedy,
\[
    \quad \forall\, l\in[m],\quad f(X_i)-f(X_{i-1})\ge \frac{1}{k_l}\cdot(f(O_{k_l})-f(X_{i-1}))\quad\textbf{($l$-th guarantee)},
\]
where $O_{k_l}$ denotes the optimal solution of cardinality $k_l$, and we call the inequality associated with $O_{k_l}$ the $l$-th guarantee. Given any $1\le l_1<l_2\le m$, if the $l_2$-th guarantee dominates (\ie, is at least as large as the $l_1$-th guarantee) at some iteration $i$, then the $l_2$-th guarantee will keep dominating the $l_1$-th guarantee for all the iterations $i'\ge i$, because the two guarantees are linear functions with variable $f(X_{i-1})$, and the $l_2$-th guarantee decreases slower than the $l_1$-th guarantee. Therefore, as $f(X_i)$ increases, the best guarantee can only transit from some $l$ to some $l'>l$. Given an instance, we let $t\le m-1$ be the number of times such transition occurs until $k_{m}$-th iteration and let $l_1<l_2<\dots<l_t$ be the indices of the corresponding best guarantees. 

For $j\le t-1$, let $F_j$ be the lowest possible value of $f(X_{i-1})$, for which the $j$-th transition occurs,
\begin{equation}\label{eq:def_F_j}
    \underbrace{\frac{1}{k_{l_j}}\cdot(f(O_{k_{l_j}})-F_j)}_{\text{($l_j$-th guarantee)}}=\underbrace{\frac{1}{k_{l_{j+1}}}\cdot(f(O_{k_{l_{j+1}}})-F_j)}_{\text{($l_{j+1}$-th guarantee)}}.
\end{equation}
We will be particularly interested in the quantity $r_j := F_j / f(O_{k_{l_j}})$.
Plugging into Eq.~\eqref{eq:def_F_j}, we have that 
\begin{equation}\label{eq:def_r_j}
    r_{j}=\frac{1-(k_{l_j}/k_{l_{j+1}})\cdot(f(O_{k_{l_{j+1}}})/f(O_{k_{l_j}}))}{1-(k_{l_j}/k_{l_{j+1}})}.
\end{equation}

\subsubsection*{Lower bounding the total value of the greedy solution recursively}

For $q\geq0$, we denote by $f^{(\textrm{greedy-lb})}(q)$ the best lower bound induced by the union of  ``$l$-th guarantees'' on the value of the $q$-th iterate of the greedy algorithm, namely
\[
\fglb(q) := \fglb(q-1)+ \max_l\left\{\frac{1}{k_l}\cdot(f(O_{k_l})-\fglb(q-1)\right\}.
\]
Now we analyze $\fglb(q)$ specifically for the instance with before-mentioned guarantee transitions. We start from the $l_1$-th guarantee and let $\fglb(0)=0$. 
Inductively, suppose that in the current iteration $q$, the $l_j$-th guarantee dominates the others, we apply the $l_j$-th guarantee $\fglb(q)-\fglb(q-1)=(f(O_{k_{l_j}})-\fglb(q-1))/k_{l_j}$ and continue iteratively until we reach some $i_j$-th iteration such that $\fglb(i_j-1)\le r_j\cdot f(O_{k_{l_j}})<\fglb(i_j)$. At the $i_j$-th iteration, $l_{j+1}$-th guarantee starts dominating, and thus, we switch to the $l_{j+1}$-th guarantee and continue like above.


\subsubsection*{Approximation ratio based on $\fglb(q)$'s is determined by $r_j$'s}
We claim that the parameters $r_j$ fully determine the ratio between the greedy lower bound $\fglb(k_{l_i})$ and $f(O_{k_{l_i}})$ for all $i\in[t]$. To see this, first observe that by~\ref{eq:def_r_j} we can infer $f(O_{k_{l_{j+1}}})/f(O_{k_{l_j}})$ from $r_j$. We can assume that $f(O_{k_{l_1}})$ is fixed without loss of generality, and then the parameters $r_j$ determine all the remaining $f(O_{k_{l_j}})$, \ie, for all $1<j\le t$,
\begin{equation}\label{eq:r_j_to_f_opt}
    f(O_{k_{l_j}})=f(O_{k_{l_1}})\cdot\prod_{j'\le j-1}\left(\frac{k_{l_{j'+1}}}{k_{l_{j'}}}-\left(\frac{k_{l_{j'+1}}}{k_{l_{j'}}}-1\right)r_{j'}\right).
\end{equation}
Moreover, the greedy lower bound $\fglb(k_{l_i})$ by definition is a linear combination of the $f(O_{k_{l_j}})$'s. Therefore, the ratio between any $\fglb(k_{l_i})$ and $f(O_{k_{l_i}})$ is fully characterized by $r_j$. In the other words, given an instance, we can get the approximation ratios of the greedy algorithm that depend only on its parameters $r_j$.

By definition of $r_1$, any feasible $r_1$ has to satisfy $r_1\le 1$, and by our assumption of the transitions, any feasible $r_j$ should satisfy $r_{j-1}\cdot f(O_{k_{l_{j-1}}})\le r_{j}\cdot f(O_{k_{l_j}})$ for all $1<j\le t$, which is equivalent to
\begin{align}\label{eq:constraint_r_j}
    r_{j-1}&\le
    r_{j}\cdot(f(O_{k_{l_j}})/f(O_{k_{l_{j-1}}})) \\
    &=
    r_{j}\cdot\left(\frac{k_{l_{j}}}{k_{l_{j-1}}}-\left(\frac{k_{l_{j}}}{k_{l_{j-1}}}-1\right)r_{j-1}\right) &&\text{(By Eq.~\eqref{eq:r_j_to_f_opt})} \nonumber\\
&    =r_{j}\cdot\left(1+\frac{k_{l_{j}}-k_{l_{j-1}}}{k_{l_{j-1}}}-\left(\frac{k_{l_{j}}-k_{l_{j-1}}}{k_{l_{j-1}}}\right)r_{j-1}\right) \nonumber\\  
& =r_{j}\cdot\left(1+\left(\frac{k_{l_{j}}-k_{l_{j-1}}}{k_{l_{j-1}}}\right)(1-r_{j-1})\right).
\end{align}
Next, for any feasible $r_j$'s (and in particular the $r_j$'s that correspond to the arbitrary instance $f^{\textrm{(bad-for-greedy)}}$ in the theorem statement), we construct a hard instance $f$ that is characterized by the same $r_j$'s (\ie, it satisfies Eq.~\eqref{eq:def_r_j} for the given $r_j$'s), such that for the hard instance $f$ and for every budget $k_i$, up to an arbitrarily small multiplicative error, (i) the aforementioned approximation ratio determined by the $r_j$'s is also an upper bound of the approximation ratio of greedy (which implies the second item in the theorem statement, because the approximation ratio determined by the particular $r_j$'s corresponding to $f^{\textrm{(bad-for-greedy)}}$ is a lower bound of the approximation ratio of greedy on $f^{\textrm{(bad-for-greedy)}}$), and (ii) greedy performs at least (almost) as good as the efficient $\mathcal{A}$ (which implies the first item in the theorem statement).


\subsubsection*{Construction of hard instance}
Let $\Delta_1=k_{l_{1}}$, $\Delta_j=k_{l_{j}}-k_{l_{j-1}},\,\forall\, 1<j<t$, and $\Delta_t=k_{m}-k_{l_{t-1}}$. We apply Theorem~\ref{thm:basic_hardness} to create $t$ hard (with respect to greedy and $\mathcal{A}$) functions $f_{\Delta_1},\dots,f_{\Delta_t}$ over disjoint ground sets $V_1,\dots,V_t$. We normalize these functions such that they have the same optimal value $1$ (\ie, $f_{\Delta_j}(O^{(j)})=1$, where $O^{(j)}$ denotes the optimal size-$\Delta_j$ solution for $f_{\Delta_j}$) and extend them to the ground set $V:=\cup_{i=1}^tV_i$.
 The final submodular function is 
 $$f(X):=\sum_{j=1}^t \alpha_j\cdot f_{\Delta_j}(X),$$ 
 where $\alpha_1:=1$ and 
 $$\alpha_j:=(\Delta_j/\sum_{s=1}^{j-1} \Delta_s)\cdot(\sum_{s=1}^{j-1}\alpha_s)\cdot(1-r_{j-1}), \;\;\;\text{for all}\; 1<j\le t.$$

\begin{claim}\label{claim:increasing_r_j}
For any $r_j$ that satisfy Eq.~\eqref{eq:constraint_r_j}, $r_j\cdot \sum_{s=1}^{j}\alpha_s$ is non-decreasing in $j$.
\end{claim}
\begin{proof}[Proof of Claim~\ref{claim:increasing_r_j}]
\begin{align*}
    r_j\cdot\sum_{s=1}^{j}\alpha_s&=r_j\cdot\left(\sum_{s=1}^{j-1}\alpha_s+\alpha_{j}\right)\\
&    =r_j\cdot\left(\sum_{s=1}^{j-1}\alpha_s+\frac{\Delta_{j}}{\sum_{s=1}^{j-1} \Delta_s}\cdot\left(\sum_{s=1}^{j-1}\alpha_s\right)\cdot(1-r_{j-1})\right) &&\text{(Definition of $\alpha_s$)}\\
    &=r_j\cdot\left(\sum_{s=1}^{j-1}\alpha_s+\frac{k_{l_{j}}-k_{l_{j-1}}}{k_{l_{j-1}}}\cdot\left(\sum_{s=1}^{j-1}\alpha_s\right)\cdot(1-r_{j-1})\right) &&\text{(Telescoping sum)}\\
    &=\left(1+\frac{k_{l_{j}}-k_{l_{j-1}}}{k_{l_{j-1}}}\cdot(1-r_{j-1})\right) r_j\cdot\sum_{s=1}^{j-1}\alpha_s\\
    &\ge r_{j-1}\cdot \sum_{s=1}^{j-1}\alpha_s, &&\text{(Ineq.~\eqref{eq:constraint_r_j})}
\end{align*}
\end{proof}

Because any feasible $r_j$'s that we need to consider satisfy $r_1\le 1$ and Eq.~\eqref{eq:constraint_r_j}, it follows by Claim~\ref{claim:increasing_r_j} and $\sum_{s=1}^{j-1}\alpha_s-(\sum_{s=1}^{j-1}\Delta_s/\Delta_j)\alpha_j=r_{j-1}\cdot \sum_{s=1}^{j-1}\alpha_s$ that $\alpha_j / \Delta_j$ is decreasing as $j$ increases. Hence, $f(O_{k_{l_{j}}})=\sum_{i=1}^j \alpha_i$ for all $j\in [t]$. Moreover, it is easy to verify that the $r_j$'s indeed characterize the $f$ constructed above in the sense that Eq.~\eqref{eq:def_r_j} holds for the $f$ constructed above.

\subsubsection*{Upper bounding greedy performance on the hard instance: a single step}
First, we analyze the best possible improvement of a single step of greedy on this instance. Suppose that greedy has chosen some size-$i$ set $X_i^{(j)}\subset V_j$, if it chooses another element from $V_j$, then we claim that the marginal gain is almost always $(f_{\Delta_j}(O^{(j)})-f_{\Delta_j}(X_i^{(j)}))/\Delta_j$ (it is at least this amount by greedy guarantee). Assume otherwise, for some $\gamma,\epsilon_1,\epsilon_2>0$, in the first $\gamma\cdot\Delta_j$ iterations when greedy chooses elements from $V_j$, there are more than $\epsilon_1\cdot \Delta_j$ iterations $i$ in which the marginal gain is larger than $((1+\epsilon_2)/\Delta_j)\cdot(f_{\Delta_j}(O^{(j)})-f_{\Delta_j}(X_i^{(j)}))$. Suppose that at $(\gamma\cdot\Delta_j)$-th iteration $f_{\Delta_j}(X_{\gamma\cdot\Delta_j}^{(j)})=c\cdot f_{\Delta_j}(O^{(j)})$, then each of those $\epsilon_1\cdot \Delta_j$ iterations gets at least an extra $(\epsilon_2/\Delta_j)\cdot(1-c)f_{\Delta_j}(O^{(j)})$ in addition to basic greedy guarantee, which implies that $f_{\Delta_j}(X_{\Delta_j}^{(j)})\ge (1-e^{-\gamma}+\epsilon_1\cdot\epsilon_2\cdot(1-c))f_{\Delta_j}(O^{(j)})$. Then, for some $\epsilon_3>0$, $f_{\Delta_j}(X_{\gamma\cdot\Delta_j}^{(j)})\ge(1-e^{-\gamma}+\epsilon_3)f_{\Delta_j}(O^{(j)})$, which is impossible by Theorem~\ref{thm:basic_hardness}. Henceforth, we can assume that the marginal gain for $f_{\Delta_j}$ is always $(f_{\Delta_j}(O^{(j)})-f_{\Delta_j}(X_i^{(j)}))/\Delta_j$ for the $i$-th iteration when greedy chooses elements from $V_j$, and this will only decrease all the values of interest by an arbitrarily small multiplicative error.

\subsubsection*{Upper bounding greedy performance on the hard instance: total value}
When we start running the greedy algorithm, for a while it only select elements from $V_1$ since those have the highest marginal contribution.
Specifically, suppose that at the beginning of the $q$-th step, greedy has selected $X_{q-1}\subset V_1$. Then the best achievable marginal gain of an element from $V_1$ for $f$ is $\alpha_1(f_{\Delta_1}(O^{(1)})-f_{\Delta_1}(X_{q-1}))/\Delta_1$. 
In comparison, the best singleton value of an element in $V_2$ is $(\alpha_2/\Delta_2)f_{\Delta_2}(O^{(2)})$, which is dominated by $\alpha_1(f_{\Delta_1}(O^{(1)})-f_{\Delta_1}(X_{q-1}))/\Delta_1$, when $f_{\Delta_1}(X_{q-1})\le r_1\cdot f_{\Delta_1}(O^{(1)})$, because
\begin{align}\label{eq:V1_dominate_V2}
    \frac{\alpha_1}{\Delta_1}(f_{\Delta_1}(O^{(1)})-f_{\Delta_1}(X_{q-1})) &\ge\frac{\alpha_1}{\Delta_1}(f_{\Delta_1}(O^{(1)})-r_1\cdot f_{\Delta_1}(O^{(1)})) && \text{(By $f_{\Delta_1}(X_{q-1})\le r_1\cdot f_{\Delta_1}(O^{(1)})$)} \nonumber \\
    &= \frac{\alpha_1}{\Delta_1}(1-r_1) && \text{(By $f_{\Delta_1}(O^{(1)})=1$)} \nonumber \\
    &= \frac{\alpha_2}{\Delta_2} && \text{(By definition of $\alpha_2$)} \nonumber \\
    &= \frac{\alpha_2}{\Delta_2}f_{\Delta_2}(O^{(2)}) && \text{(By $f_{\Delta_2}(O^{(2)})=1$)}.
\end{align}
Thus, when $f(X_{q-1})=\alpha_1\cdot f_{\Delta_1}(X_{q-1})<r_1\cdot f(O^{(1)})$, greedy should always prefer choosing elements from $V_1$ over $V_2$ (and other $V_i$'s), and the single step improvement is $f(X_q)-f(X_{q-1})=(f(O^{(1)})-f(X_{q-1}))/\Delta_1=(f(O^{(1)})-f(X_{q-1}))/k_{1}$ (this matches how $\fglb(q)$ changes).

We now analyze what happens when, after running greedy for a while,  the marginal contribution from $V_1$-elements decays so that greedy may prefer $V_2$-elements. 
By Eq.~\eqref{eq:V1_dominate_V2}, it is when $f_{\Delta_1}(X_{q-1})= r_1\cdot f_{\Delta_1}(O^{(1)})$ that the best singleton value of $V_2$-elements $(\alpha_2/\Delta_2)f_{\Delta_2}(O^{(2)})$ becomes equal to the best marginal contribution of a $V_1$-element $\alpha_1(f_{\Delta_1}(O^{(1)})-f_{\Delta_1}(X_{q-1}))/\Delta_1$.
Therefore, once $r_2\cdot f(O^{(1)}\cup O^{(2)})> f(X_{q-1})\ge r_1\cdot f(O^{(1)})$, greedy should start choosing elements from $V_1$ and $V_2$ alternatively to keep the identity $\alpha_1(f_{\Delta_1}(O^{(1)})-f_{\Delta_1}(X_{q-1}))/\Delta_1=\alpha_2(f_{\Delta_2}(O^{(2)})-f_{\Delta_2}(X_{q-1}))/\Delta_2$ (up to negligible error), it follows that
\begin{align*}
    \frac{\alpha_1(f_{\Delta_1}(O^{(1)})-f_{\Delta_1}(X_{q-1}))}{\Delta_1}&=\frac{\alpha_2(f_{\Delta_2}(O^{(2)})-f_{\Delta_2}(X_{q-1}))}{\Delta_2} \\
    &=\frac{\alpha_1(f_{\Delta_1}(O^{(1)})-f_{\Delta_1}(X_{q-1}))+\alpha_2(f_{\Delta_2}(O^{(2)})-f_{\Delta_2}(X_{q-1}))}{\Delta_1+\Delta_2} \\
    &=\frac{f(O^{(1)}\cup O^{(2)})-f(X_{q-1})}{k_{l_2}}.
\end{align*}
Thus, there is a transition of the best marginal gain when $f(X_{q-1})=r_1\cdot f(O^{(1)})$ with $X_{q-1}\subseteq V_1$, and after that the best achievable marginal gain is characterized by $(f(O^{(1)}\cup O^{(2)})-f(X_{q-1}))/k_{l_2}$ (this matches the guarantee transition for $\fglb(q)$), which is larger than $(\alpha_3/\Delta_{3})f_{\Delta_3}(O^{(3)})$ by definition of $\alpha_3$.

Similarly, for every $3\le p\le t$, when $f(X_{q-1})=r_{p-1}\cdot f(\cup_{j=1}^{p-1}O^{(j)})$ with $X_{q-1}\subseteq\cup_{j\le p-1} V_j$, it holds that $(f(\cup_{j=1}^{p-1}O^{(j)})-f(X_{q-1}))/k_{l_{p-1}}=(\alpha_p/\Delta_p)f_{\Delta_p}(O^{(p)})$ by definition of $\alpha_p$, and hence, greedy starts to choose elements from $V_1,\dots,V_p$ to keep $\alpha_j(f_{\Delta_j}(O^{(j)})-f_{\Delta_j}(X_{q-1}))/\Delta_j$ for all $j\le p$ approximately equal to each other. Hence, for all $j\le p$,
\[
    \frac{\alpha_j(f_{\Delta_j}(O^{(j)})-f_{\Delta_j}(X_{q-1}))}{\Delta_j}=\frac{\sum_{j\le p}\alpha_j(f_{\Delta_j}(O^{(j)})-f_{\Delta_j}(X_{q-1}))}{\sum_{j\le p}\Delta_j}=\frac{f(\cup_{j=1}^{p}O^{(j)})-f(X_{q-1})}{k_{l_{p}}},
\]
and this is a transition of the best marginal gain from $(f(\cup_{j=1}^{p-1}O^{(j)})-f(X_{q-1}))/k_{l_{p-1}}$ to $(f(\cup_{j=1}^{p}O^{(j)})-f(X_{q-1}))/k_{l_{p}}$ (this matches the guarantee transition for $\fglb(q)$). Therefore, we have shown that the greedy performance $f(X_q)$ changes in exactly the same way as $\fglb(q)$, and hence, the approximation ratio based on $\fglb(q)$'s is tight for greedy on the hard instance.

\subsubsection*{How greedy spends the budget} Finally, following the above derivation, we emphasize how greedy spends the budget. As we have shown, for any $1\le p\le t$, when $r_{p-1}\cdot f(\cup_{j=1}^{p}O^{(j)})\le f(X_{q-1})\le r_{p}\cdot f(\cup_{j=1}^{p}O^{(j)})$, greedy splits its budget on $V_1,\dots,V_p$ to keep all the $\alpha_j(f_{\Delta_j}(O^{(j)})-f_{\Delta_j}(X_{q-1}))/\Delta_j$ approximately equal to each other. Moreover, for any $j'>p$, the best singleton value $(\alpha_{j'}/\Delta_{j'})f_{\Delta_{j'}}(O^{(j')})$ of $V_{j'}$ is smaller than $\alpha_j(f_{\Delta_j}(O^{(j)})-f_{\Delta_j}(X_{q-1}))/\Delta_j$ for any $j\le p$. Suppose that greedy has spent budget $\hat{b}_j$ on $V_j$ for each $j$, which implies that $f_{\Delta_j}(X_{q-1})=1-e^{-\hat{b}_j/\Delta_j}$. Then, we have that $\alpha_j(f_{\Delta_j}(O^{(j)})-f_{\Delta_j}(X_{q-1}))/\Delta_j=\frac{d\alpha_j(1-e^{-x/\Delta_j})}{dx}|_{x=\hat{b}_j}$, and thus, $\frac{d\alpha_j(1-e^{-x/\Delta_j})}{dx}|_{x=\hat{b}_j}$ for all $j\le p$ are equal to each other. Moreover, since $(\alpha_{j'}/\Delta_{j'})f_{\Delta_{j'}}(O^{(j')})=\frac{d\alpha_{j'}(1-e^{-x/\Delta_{j'}})}{dx}|_{x=0}$, $\frac{d\alpha_j(1-e^{-x/\Delta_j})}{dx}|_{x=\hat{b}_j}\ge\frac{d\alpha_{j'}(1-e^{-x/\Delta_{j'}})}{dx}|_{x=0}$ for all $j\le p$ and $j'>p$.

\subsubsection*{Greedy spends the budget optimally on the hard instance}
By Theorem~\ref{thm:basic_hardness} and the design of our hard instances, for any budget $b_j$, the best possible value the efficient algorithm $\mathcal{A}$ can get by spending budget $b_j$ on $f_{\Delta_j}$ is $u_j(b)=(1-e^{-b_j/\Delta_j})\alpha_j$. Suppose $\mathcal{A}$ spends budget $b_j$ on each $f_{\Delta_j}$, where $\sum_{j=1}^{t}b_j=b$ for some $b$, then in this case, the best possible value in total is $\sum_{j=1}^{t}u_j(b_j)$, and hence, in general, the best possible value for budget $b$ is upper bounded by the maximum of the following program:
\[  
    \max \sum_{j=1}^{t}u_j(b_j)\quad\textrm{s.t. }\sum_{j=1}^{t}b_j=b \;\; \textrm{and} \; \forall j \; b_j\ge 0.
\]
We observe that for an arbitrary fixed $b$, the maximizer $b_j^*$'s for this program should satisfy that for all positive $b_j^*$, the derivatives of $u_j$'s at $b_j^*$'s are equal (notice that the way greedy spends the budget also satisfies this property), and moreover, they are not smaller than the derivative of $u_{j'}$'s at $0$ for any $j'$ such that $b_{j'}^*=0$. Otherwise, there must exist $\frac{du_{j_1}(x)}{dx}|_{x=b_{j_1}^*}<\frac{du_{j_2}(x)}{dx}|_{x=b_{j_2}^*}$ where $b_{j_1}^*$ is strictly positive, then increasing $b_{j_2}^*$ by $\delta$ and decreasing $b_{j_1}^*$ by $\delta$ for sufficiently small $\delta$ will increase the objective value while preserving the feasibility of $b_j^*$'s.

Now we prove that for any fixed $b$, the $b_j^*$'s satisfying the above mentioned property are unique. (Then, it follows that the maximizer matches exactly how greedy spends the budget, and moreover, greedy attains the optimal value of the program.) Suppose that besides $b_j^*$'s, $\widetilde{b}_j$'s also satisfy the property. Let $\textrm{supp}(b^*)$ be the set of $j$ such that $b_j^*>0$. We first argue if $j'\notin \textrm{supp}(b^*)$, then $\widetilde{b}_{j'}=0$. Suppose otherwise, $\widetilde{b}_{j'}>0$, then $\sum_{j\in\textrm{supp}(b^*)}\widetilde{b}_j<t$, and hence, there must exist a $j\in\textrm{supp}(b^*)$ such that $\widetilde{b}_j<b_j^*$. By strict concavity of $u_j$, $\frac{du_{j}(x)}{dx}|_{x=\widetilde{b}_{j'}}>\frac{du_{j}(x)}{dx}|_{x=b_{j}^*}$. However, since the $b_j^*$'s satisfy the above mentioned property and $b_{j'}^*=0$, $\frac{du_{j}(x)}{dx}|_{x=b_{j}^*}\ge \frac{du_{j'}(x)}{dx}|_{x=0}$, and by strict concavity of $u_{j'}$, $\frac{du_{j'}(x)}{dx}|_{x=0}> \frac{du_{j'}(x)}{dx}|_{x=\widetilde{b}_{j'}}$, which gives a contradiction. Furthermore, we can argue that for all $j\in\textrm{supp}(b^*)$, $\widetilde{b}_{j'}=b^*_{j}$, because otherwise, there must exist $j,j'\in\textrm{supp}(b^*)$ such that $b^*_{j}>\widetilde{b}_{j}$ and $b^*_{j'}<\widetilde{b}_{j'}$, and hence $\frac{du_{j}(x)}{dx}|_{x=b^*_{j}}<\frac{du_{j}(x)}{dx}|_{x=\widetilde{b}_{j}}$ and $\frac{du_{j'}(x)}{dx}|_{x=b^*_{j'}}> \frac{du_{j'}(x)}{dx}|_{x=\widetilde{b}_{j'}}$, which contradicts the property $\frac{du_{j}(x)}{dx}|_{x=\widetilde{b}_{j}}=\frac{du_{j'}(x)}{dx}|_{x=\widetilde{b}_{j'}}$.


\end{proof}

\subsection{Two Remarks for Theorem~\ref{thm:many_budgets}}
One might wonder whether the transitions of the greedy guarantees in the above analysis of Theorem~\ref{thm:many_budgets} always occur in the order $1,2,\dots,m$ but never in any proper subsequence, namely whether $r_l\cdot f(O_{k_{l}})\le r_{l-1}\cdot f(O_{k_{l-1}})$, which is equivalent to
\[
    \frac{f(O_{k_{l}})-(k_{l}/k_{l+1})f(O_{k_{l+1}})}{1-(k_{l}/k_{l+1})}\ge \frac{f(O_{k_{l-1}})-(k_{l-1}/k_{l})f(O_{k_{l}})}{1-(k_{l-1}/k_{l})},
\]
This is equivalent to
\[
    f(O_{k_{l+1}})-f(O_{k_{l}})\le\frac{k_{l+1}-k_l}{k_l-k_{l-1}}\cdot(f(O_{k_{l}})-f(O_{k_{l-1}})),
\]
which is actually true for our instances but not in general. See Example~\ref{ex:increasing_opt_diff}.

\begin{example}\label{ex:increasing_opt_diff}
Consider the function $f$ on the ground set $\{1,2,3,4\}$ with values $f(\emptyset)=0,\,f(\{1\})=1,\,f(\{2\})=f(\{3\})=f(\{4\})=1/2$, $f(\{1,2\})=f(\{1,3\})=f(\{1,4\})=7/6$, $f(\{2,3\})=f(\{2,4\})=f(\{3,4\})=1$, $f(\{2,3,4\})=3/2$, $f(\{1,2,3\})=f(\{1,2,4\})=f(\{1,3,4\})=4/3$ and $f(\{1,2,3,4\})=3/2$. It is straightforward to check that $f$ is submodular and monotone, and that $f(O_3)-f(O_2)>f(O_2)-f(O_1)$.
\end{example}

Finally, we end this section with the following observation. In the appendix, we give the proof of this observation and show that many practical algorithms satisfy the condition of this observation.
\begin{observation}\label{obs:opt_algorithms}
For any perturbation factors $0<\rho_1<\rho_2<\dots<\rho_{m}$, there exists a sufficiently large $k$ that grows with the size of instance such that given $m$ budgets $k_1=\rho_1\cdot k$, \dots, $k_{m}=\rho_{m}\cdot k$, the optimality described in Theorem~\ref{thm:many_budgets} actually holds for a general class of algorithms such that:
\begin{itemize}
    \item Given budget $k_i$, the algorithm $\mathcal{A}$ runs in $T$ rounds ($T$ is sufficiently large), of which each round selects about $k_i/T$ elements.
    \item For any $\epsilon>0$, it holds for all $t\in[T]$, for all $j\in[m]$, that $f(X^{\mathcal{A}}_{tk_i/T})-f(X^{\mathcal{A}}_{(t-1)k_i/T})\ge ((1-\epsilon)\rho_i/(\rho_j T))\cdot(f(O_{k_j})-f(X^{\mathcal{A}}_{tk_i/T}))$, where $X^{\mathcal{A}}_s$ is the $s$-th element chosen by $\mathcal{A}$.
\end{itemize}
\end{observation}

\section{A Lower Bound for Every Distribution}\label{sec:bound}

The main thesis of this paper is that worst case instances of submodular maximization are really tailored to a specific budget constraint. 
It is natural to hope that as the distribution of budget perturbation becomes arbitrarily spread (aka arbitrarily far from the worst case single budget), the approximation factor approaches $1$. 
%
In this section, we give a negative answer to this question. 

\begin{theorem}\label{thm:lower_bound}
For any distribution of budget perturbations $\mathcal{D}$, for any efficient algorithm $\mathcal{A}$, for every sufficiently large $k$ that grows with the size of instance, $\cR_{\mathcal{A}}(\mathcal{D}(k))\le 0.9087$.
\end{theorem}

(We did not seriously try to optimize the constant $0.9087$. Computing the optimal constant is an interesting open problem for future work.)

\begin{proof}
For arbitrarily small $\tau>0$, let $\rho_{\min}$ and $\rho_{\max}$ be such that the mass of $\mathcal{D}$ on $[\rho_{\min}, \rho_{\max}]$ is at least $1-\tau$.
Let $q=50$. Let $K_1=q^{-(i^*-1)}\cdot k$ where $i^*$ is the largest $i$ such that $q^{-(i-1)}\cdot k\le \rho_{\min}\cdot k$. Let $N$ be the smallest $i$ such that $q^{(i-1)}\cdot K_1\ge \rho_{\max}\cdot k$. We first construct the hard instances, and by Theorem~\ref{thm:many_budgets} it suffices to upper bound the approximation ratio achieved by greedy on these instances.

\subsubsection*{Construction of hard instances}
Let $K_{i}=q^{(i-1)}\cdot K_1$. We use Theorem~\ref{thm:basic_hardness} to create hard (with respect to greedy algorithm) functions $f_{K_{i}}$ for all $i\in [N]$ over disjoint ground sets. We normalize these functions such that they have the same optimal value $1$ (\ie, $f_{K_i}(O^{(i)})=1$, where $O^{(i)}$ denotes the optimal size-$K_i$ solution for $f_{K_i}$) and extend them to the union of all the ground sets. The final submodular function is $f(X)=\sum_{i=1}^N (q/e)^{i-1}\cdot f_{K_i}(X)$.

\subsubsection*{Upper bounding the approximation ratio on the hard instances}
Consider a budget $K$ between $\sum_{j=1}^i K_{j}$ and $\sum_{j=1}^{i+1} K_{j}$, for any $i\le N-1$. We first show that the contribution of $f_{K_j}$ with $j\le i-1$ is negligible. Notice that the best singleton value of $f_{K_j}$ is $(q/e)^{j-1}/(q^{j-1}\cdot K_1)$, which is decreasing in $j$. Hence, we can generously assume that the algorithm spends a budget of size $\sum_{j=1}^{i-1} K_{j}$ getting all the utilities from $f_{K_j}$ with $j\le i-1$, which is the best one can hope for. The total value of these $f_{K_j}$'s is $\sum_{j=1}^{i-1}(q/e)^{j-1}=((q/e)^{i-1}-1)/(q/e-1)$, which is less than $1/(q/e-1)<0.0575$ fraction of the value of the $f_{K_i}$. Therefore, the best possible approximation ratio for budget $K$ is at most the best possible approximation ratio for budget $K-\sum_{j=1}^{i-1} K_{j}$ on the $f_{K_j}$'s with $j\ge i$ plus $0.0575$.

Furthermore, the best singleton value of $f_{K_{i+1}}$ is at most $(q/e)^{i}/(q^{i}\cdot K_1)$. On the other hand, with budget $K_i$ on $f_{K_i}$, greedy can achieve approximation ratio at most $1-1/e$ by Theorem~\ref{thm:basic_hardness}, and thus, at the $(K_i+1)$-th iteration, greedy has marginal gain at least $(1-(1-1/e))\cdot(q/e)^{i-1}/(q^{(i-1)}\cdot K_1)$, which is equal to the best singleton value of $f_{K_{i+1}}$. Hence, greedy will not choose anything from $f_{K_{i+1}}$ until it has selected $K_i$ elements from $f_{K_i}$. The remaining budget $K-\sum_{j=1}^{i} K_{j}$ is at most $K_{i+1}$, and it follows for the same reason that greedy will not spend remaining budget on any $f_{K_j}$'s with $j\ge i+2$.

It remains to show how greedy performs on $f_{K_i}$ and $f_{K_{i+1}}$ with budget $K'=K-\sum_{j=1}^{i} K_{j}$. Let $a=K'/K_i$. Notice that greedy splits its budget in the way that the marginal gain of choosing the next element from $f_{K_i}$ is approximately equal to that of choosing the next element from $f_{K_{i+1}}$. This can be expressed as the following equations:
\begin{align*}
    a_1K_i+a_2K_i&=aK_i, \\
    e^{-a_1}\cdot \frac{e/q}{K_i}&=e^{-\frac{a_2K_i}{K_{i+1}}}\cdot \frac{1}{K_{i+1}},
\end{align*}
where $a_1K_i$ and $a_2K_i$ are the budgets spent on $f_{K_i}$ and $f_{K_{i+1}}$ respectively. The solution is $a_1=(a+q)/(q+1)$ and $a2=(q(a-1))/(q+1)$. Hence the approximation ratio of greedy on $f_{K_i}$ and $f_{K_{i+1}}$ with budget $K'$ is at most
\[
    \frac{(1-e^{-a_1K_i/K_i})\cdot\frac{e}{q}+(1-e^{-a_2K_i/K_{i+1}})}{\frac{e}{q}+\frac{aK_i-K_i}{K_{i+1}}}=\frac{(1-e^{-(a+q)/(q+1)})\cdot\frac{e}{q}+(1-e^{-(a-1)/(q+1)})}{\frac{e}{q}+\frac{a-1}{q}}.
\]
The maximum is approximately $0.85114$ achieved by $a\approx 9.2199$. Hence, in total, the approximation ratio for the entire instance is less than $0.8512+0.0575=0.9087$. We have proved that the optimal approximation ratio for the hard instance is less than $0.9087$ for any budget in $[\rho_{\min}\cdot k, \rho_{\max}\cdot k]$, and this finishes the proof because $\tau$ is arbitrarily small.

\end{proof}


\section{Numerical Simulation}\label{sec:simulation}
We formulate a mathematical program that computes the worst possible optimal expected approximation ratio, for any fixed distribution on any fixed choice of $m$ budgets $\rho_1 k<\rho_2 k<\dots<\rho_m k=k$ (we also let $\rho_0=0$). We denote the probability of budget $\rho_i k$ by $p_i$ for each $i$.


\subsubsection*{Reducing the hard instances to a standard form}
Recall that the hard instances in the proof of Theorem~\ref{thm:many_budgets} have following form---$f^*(X)=\sum_{j=1}^t\alpha_j\cdot f_{(\rho_{l_j}-\rho_{l_{j-1}})k}$, where $l_1,\dots,l_t$ is a subsequence of $1,\dots,m$, and $f_{(\rho_{l_j}-\rho_{l_{j-1}})k}(X)$ is the hard submodular function from Theorem~\ref{thm:basic_hardness}, and it is normalized such that its optimal value for budget $(\rho_{l_j}-\rho_{l_{j-1}})k$ is $1$. Moreover, $\alpha_j/(\rho_{l_j}-\rho_{l_{j-1}})$ is increasing in $j$.
We show that there is a submodular function that is as hard as $f^*$ to approximately maximize in the following \textbf{standard form}\footnote{The difference between $f^*$ and the standard form $f$ is that in the $f$ there is a sub-instance for every budget}---$f(X)=\sum_{i=1}^m\beta_i\cdot f_{(\rho_i-\rho_{i-1})k}(X)$, where $\beta_i$'s satisfy
\begin{equation}\label{eq:program_constraint}
    \frac{\beta_i}{\rho_i-\rho_{i-1}}\ge\frac{\beta_{i+1}}{\rho_{i+1}-\rho_i},\quad\forall\, i<m,
\end{equation}
where $f_{(\rho_i-\rho_{i-1})k}$ is defined analogously to $f_{(\rho_{l_j}-\rho_{l_{j-1}})k}$, and we denote the ground set of $f_{(\rho_i-\rho_{i-1})k}$ by $V_i$. 

For $i \in [m]$ and $j$ such that $l_{j-1}< i\le l_j$, we define $$\lambda_i:=\frac{\rho_i-\rho_{i-1}}{\rho_{l_j}-\rho_{l_{j-1}}}.$$ 
\begin{claim}
Given budget $x\cdot k$ for any $x\ge 0$, the best achievable approximation ratio for $\sum_{i=l_{j-1}+1}^{l_j}\lambda_i\cdot f_{(\rho_{i}-\rho_{i-1})k}(X)$ is equal to that for $f_{(\rho_{l_j}-\rho_{l_{j-1}})k}(X)$.
\end{claim}
\begin{proof}
 For any budget $x\cdot k$, the best achievable approximation ratio for $\sum_{i=l_{j-1}+1}^{l_j}\lambda_i\cdot f_{(\rho_{i}-\rho_{i-1})k}(X)$ is
\[
    \max_{x_i\textrm{'s}}\sum_{i=l_{j-1}+1}^{l_j}\lambda_i(1-e^{-\frac{x_i}{\rho_i-\rho_{i-1}}})\quad\textrm{s.t. $\sum_{i=l_{j-1}+1}^{l_j}x_i=x$ and $x_i$'s are non-negative}.
\]
For any feasible $x_i$'s,
\begin{align*}
    \sum_{i=l_{j-1}+1}^{l_j}\lambda_i(1-e^{-\frac{x_i}{\rho_i-\rho_{i-1}}})&\le 1-e^{-\sum_{i=l_{j-1}+1}^{l_j}\frac{\lambda_i\cdot x_i}{\rho_i-\rho_{i-1}}} && \text{(Jensen's inequality and $\sum_{i=l_{j-1}+1}^{l_j}\lambda_i=1$)} \\
    &=1-e^{-\sum_{i=l_{j-1}+1}^{l_j}\frac{x_i}{\rho_{l_j}-\rho_{l_{j-1}}}} && \text{(By definition of $\lambda_i$)} \\
    &=1-e^{-\frac{x}{\rho_{l_j}-\rho_{l_{j-1}}}} && \text{(By $\sum_{i=l_{j-1}+1}^{l_j}x_i=x$)}.
\end{align*}
Moreover, when $\frac{x_i}{\rho_i-\rho_{i-1}}$ for all $i$ are equal to each other, we have that $\frac{x_i}{\rho_i-\rho_{i-1}}=\frac{\sum_{i=l_{j-1}+1}^{l_j} x_i}{\sum_{i=l_{j-1}+1}^{l_j} \rho_i-\rho_{i-1}}=\frac{x}{\rho_{l_j}-\rho_{l_{j-1}}}$, and then $\sum_{i=l_{j-1}+1}^{l_j}\lambda_i(1-e^{-\frac{x_i}{\rho_i-\rho_{i-1}}})=1-e^{-\frac{x}{\rho_{l_j}-\rho_{l_{j-1}}}}$. Hence, $1-e^{-\frac{x}{\rho_{l_j}-\rho_{l_{j-1}}}}$ is exactly the best achievable approximation ratio for $\sum_{i=l_{j-1}+1}^{l_j}\lambda_i\cdot f_{(\rho_{i}-\rho_{i-1})k}(X)$. Notice that it is also the best achievable approximation ratio for $f_{(\rho_{l_j}-\rho_{l_{j-1}})k}(X)$. 
\end{proof}
Henceforth, we can replace each $f_{(\rho_{l_j}-\rho_{l_{j-1}})k}$ with $\sum_{i=l_{j-1}+1}^{l_j}\lambda_i\cdot f_{(\rho_{i}-\rho_{i-1})k}(X)$ in $f^*$, which reduces $f^*$ to the standard form. Then, Eq.~\eqref{eq:program_constraint} follows by definition of $\lambda_i$'s and the monotonicity of $\alpha_j/(\rho_{l_j}-\rho_{l_{j-1}})$. Finally, we note that Eq.~\eqref{eq:program_constraint} implies that optimal value of the $f^*$ for budget $\rho_ik$ is $\sum_{j=1}^i\beta_j$ and that for any $i$, whenever $V_i$ is used by the greedy algorithm, so should the $V_{i'}$'s for any $i'\le i$.

\subsubsection*{Formulating the mathematical program}
For each budget $\rho_i\cdot k$, the best possible approximation ratio is achieved by choosing elements from the first $l$ subsets $V_1,\dots,V_l$ for certain $l\le m$ (which we do not know a priori), 
and the budget should be split in a way such that the marginal contribution from the next element is (approximately) equal among $V_1,\dots,V_l$. That is, 
\begin{align*}
    \frac{d(\beta_1(1-e^{-x_1^{(i,l)}/\rho_1}))}{dx_1^{(i,l)}}&=\frac{d(\beta_j(1-e^{-x_j^{(i,l)}/(\rho_j-\rho_{j-1})}))}{dx_j^{(i,l)}},\,\,\forall\, j\le l,\\
    \sum_{j\le l} x_j^{(i,l)}&=\rho_i, \textrm{ and } x^{(i,l)}_j\ge 0,\,\,\forall j\le l.
\end{align*}
Solving the system of equations in the above constraint gives us
\begin{equation}\label{eq:solution_balance_equation}
\begin{split}
&\frac{x_1^{(i,l)}}{\rho_1}=\frac{\rho_i}{\rho_l}-\sum_{j=1}^l \ln\left(\frac{\beta_j\rho_1}{\beta_1(\rho_j-\rho_{j-1})}\right)\cdot \frac{\rho_j-\rho_{j-1}}{\rho_l},\\
&\frac{x_j^{(i,l)}}{\rho_j-\rho_{j-1}}=\frac{x_1^{(i,l)}}{\rho_1}+\ln\left(\frac{\beta_j\rho_1}{\beta_1(\rho_j-\rho_{j-1})}\right),\,\,\forall\, j\le l.
\end{split}
\end{equation}

We let $h^{(i,l)}(\beta_1,\dots,\beta_m)$ denote the approximation ratio achieved by $x^{(i,l)}_j$'s, then it is given by
\begin{align*}
    h^{(i,l)}(\beta_1,\dots,\beta_m)&=\frac{\sum_{j=1}^l\beta_j(1-e^{-\frac{x^{(i,l)}_j}{\rho_j-\rho_{j-1}}})}{\sum_{j=1}^i \beta_j} \\
    &=\frac{\sum_{j=1}^l\beta_j(1-e^{-\frac{x_1^{(i,l)}}{\rho_1}}\cdot\frac{\beta_1(\rho_j-\rho_{j-1})}{\beta_j\rho_1})}{\sum_{j=1}^i \beta_j} && \text{(Solution of $x^{(i,l)}_j$)} \\
    &=\frac{\sum_{j=1}^l\beta_j-\sum_{j=1}^l\frac{\beta_1(\rho_j-\rho_{j-1})}{\rho_1}\cdot e^{-\frac{x_1^{(i,l)}}{\rho_1}}}{\sum_{j=1}^i \beta_j} \\
    &=\frac{\left(\sum_{j=1}^l\beta_j\right)-\beta_1\cdot\frac{\rho_l}{\rho_1}\cdot e^{-\frac{x_1^{(i,l)}}{\rho_1}}}{\sum_{j=1}^i \beta_j} && \text{(Telescoping sum)} \\
    &=\frac{\left(\sum_{j=1}^{l}\beta_j\right)-\beta_1\cdot\frac{\rho_l}{\rho_1}\cdot e^{-\frac{\rho_i}{\rho_l}+\sum_{j=1}^{l}\ln\left(\frac{\beta_j\rho_1}{\beta_1(\rho_j-\rho_{j-1})}\right)\cdot\frac{\rho_j-\rho_{j-1}}{\rho_l}}}{\sum_{j=1}^i \beta_j} && \text{(Solution of $x^{(i,l)}_1$)},\\
\end{align*}
where the nominator is the value achieved by $x^{(i,l)}_j$'s, and the denominator is the optimal value. Since we do not know the right choice of $l$ a priori, we will enumerate all possible choices of $l$ and pick the best. Moreover, for every $l\le m$, we consider $l$ as a candidate choice only if the solutions of $x^{(i,l)}_j$'s by Eq.~\eqref{eq:solution_balance_equation} are non-negative, because this holds for the right choice of $l$. (Note that $l=1$ is always a candidate choice, because it means that all budget are spent on the first sub-instance, and hence $x^{(i,1)}_1=\rho_i\ge 0$.) Therefore, we let $h^{(i)}$ denote the best approximation ratio for budget $\rho_i k$, then it is given by
\begin{align*}
    h^{(i)}(\beta_1,\dots,\beta_m)&=\max_{1\le l\le m} h^{(i,l)}(\beta_1,\dots,\beta_m)\cdot \mathds{1}[x^{(i,l)}_l\ge 0] \\
    &=\max_{1\le l\le m} h^{(i,l)}(\beta_1,\dots,\beta_m)-C\cdot\mathds{1}[x^{(i,l)}_l<0] && \text{($C$ is a large constant)},
\end{align*}
where $x^{(i,l)}_l$ can be represented as a function of $\beta_i$'s. Note that we only restrict $x^{(i,l)}_l$ to be non-negative, which actually implies that every $x^{(i,l)}_j$ for $j\le l$ is non-negative, by Eq.\eqref{eq:program_constraint} and Eq.\eqref{eq:solution_balance_equation}.
Finally, the expected approximation ratio $h$ is given by $h(\beta_1,\dots,\beta_m)=\sum_{i=1}^m p_i\cdot h^{(i)}(\beta_1,\dots,\beta_m)$. Given any fixed $\rho_i$'s, the worst possible optimal average approximation ratio is the result of the following program
\[
    \min_{\beta_1,\dots,\beta_m\ge0} h(\beta_1,\dots,\beta_m) \textrm{ s.t. Eq.\eqref{eq:program_constraint} and Eq.\eqref{eq:solution_balance_equation}}.
\]

\subsection{Empirical results}\label{sub:numerical}
We solve this program numerically for various distributions of budget perturbations ($\rho_i$'s); the results are summarized in Table~\ref{table:simulation_result}.

\begin{table}
\begin{center}
\caption{\label{table:democrat_budget} Campaign Budgets (in millions)}
\begin{tabular}{ |c|c|c|c|c|c|c|c|c|c|c| } 
 \hline
 \textbf{Candidate} & Bennet & Biden & Bloomberg & Buttigieg & Gabbard \\
 \hline
 \textbf{Budget} & 2.6 & 23.3	& 188.4	& 34.1 & 2.9 \\
 \hline
 Klobuchar & Patrick & Sanders & Steyer & Warren & Yang \\
 \hline
 10.1	& 0.9 & 50.1 & 153.7 & 33.7 & 19.2 \\
 \hline
\end{tabular}
\end{center}
\end{table}

\begin{description}
\item[Canonical distributions] It is natural to ask what is the expected approximation factor when the budget is drawn from uniform over $[x,10x]$. Since we don't know how to compute this value exactly, we take discretization of this distribution namely 25 budgets%
\footnote{We use a somewhat sparse discretization since the program is non-convex.} evenly spaced between $x$ and $10x$. Similarly, we experiment with discretizations of log-scale uniform distributions over $[x,10x]$ and $[x,600x]$.
\item[Top social/political campaigns on Facebook] With the application of influence maximization on social networks in mind, we use the budgets of the top ten campaigns on Facebook's database of social/political campaigns%
\footnote{Top ten amount spent during the month before Mar. 26, 2020 in the Facebook Ad report~\cite{2020facebook}.}.
\item[2020 Democratic Party presidential candidates] We use reported total campaign budgets by candidates in the 2020 Democratic Party primary elections during months October-December 2019~\cite{2020election} (see Table~\ref{table:democrat_budget}).%
\end{description}

\medskip
 
\bibliographystyle{alpha}
\bibliography{MasterBib}

\newcommand{\etalchar}[1]{$^{#1}$}
\begin{thebibliography}{ZWG{\etalchar{+}}17}

\bibitem[AGN18]{anari2018budget}
Nima Anari, Gagan Goel, and Afshin Nikzad.
\newblock Budget feasible procurement auctions.
\newblock {\em Operations Research}, 66(3):637--652, 2018.

\bibitem[AKS19]{AmanatidisKS19}
Georgios Amanatidis, Pieter Kleer, and Guido Sch{\"{a}}fer.
\newblock Budget-feasible mechanism design for non-monotone submodular
  objectives: Offline and online.
\newblock In {\em Proceedings of the 2019 {ACM} Conference on Economics and
  Computation, {EC} 2019, Phoenix, AZ, USA, June 24-28, 2019}, pages 901--919,
  2019.

\bibitem[AKS20]{2020election}
Sarah Almukhtar, Thomas Kaplan, and Rachel Shorey.
\newblock 2020 democrats went on a spending spree in the final months of 2019.
\newblock
  \url{https://www.nytimes.com/interactive/2020/02/01/us/elections/democratic-q4-fundraising.html},
  2020.

\bibitem[BCIW12]{BalcanCIW12}
Maria{-}Florina Balcan, Florin Constantin, Satoru Iwata, and Lei Wang.
\newblock Learning valuation functions.
\newblock In {\em {COLT} 2012 - The 25th Annual Conference on Learning Theory,
  June 25-27, 2012, Edinburgh, Scotland}, pages 4.1--4.24, 2012.

\bibitem[BDF{\etalchar{+}}12]{BadanidiyuruDFKNR12}
Ashwinkumar Badanidiyuru, Shahar Dobzinski, Hu~Fu, Robert Kleinberg, Noam
  Nisan, and Tim Roughgarden.
\newblock Sketching valuation functions.
\newblock In {\em Proceedings of the Twenty-Third Annual {ACM-SIAM} Symposium
  on Discrete Algorithms, {SODA} 2012, Kyoto, Japan, January 17-19, 2012},
  pages 1025--1035, 2012.

\bibitem[BH16]{BalkanskiH16}
Eric Balkanski and Jason~D. Hartline.
\newblock Bayesian budget feasibility with posted pricing.
\newblock In {\em Proceedings of the 25th International Conference on World
  Wide Web, {WWW} 2016, Montreal, Canada, April 11 - 15, 2016}, pages 189--203,
  2016.

\bibitem[BKS12]{BadanidiyuruKS12}
Ashwinkumar Badanidiyuru, Robert Kleinberg, and Yaron Singer.
\newblock Learning on a budget: posted price mechanisms for online procurement.
\newblock In {\em Proceedings of the 13th {ACM} Conference on Electronic
  Commerce, {EC} 2012, Valencia, Spain, June 4-8, 2012}, pages 128--145, 2012.

\bibitem[BQS21]{BQS21}
Eric Balkanski, Sharon Qian, and Yaron Singer.
\newblock Instance specific approximations for submodular maximization.
\newblock In {\em Proceedings of the 38th International Conference on Machine
  Learning}, volume 139, pages 609--618. PMLR, 2021.

\bibitem[BRS16]{BRS16}
Eric Balkanski, Aviad Rubinstein, and Yaron Singer.
\newblock The power of optimization from samples.
\newblock In {\em Advances in Neural Information Processing Systems 29: Annual
  Conference on Neural Information Processing Systems 2016, December 5-10,
  2016, Barcelona, Spain}, pages 4017--4025, 2016.

\bibitem[BRS19]{BRS19a}
Eric Balkanski, Aviad Rubinstein, and Yaron Singer.
\newblock An exponential speedup in parallel running time for submodular
  maximization without loss in approximation.
\newblock In {\em Proceedings of the Thirtieth Annual {ACM-SIAM} Symposium on
  Discrete Algorithms, {SODA} 2019, San Diego, California, USA, January 6-9,
  2019}, pages 283--302, 2019.

\bibitem[CC84]{CC84}
Michele Conforti and G{\'{e}}rard Cornu{\'{e}}jols.
\newblock Submodular set functions, matroids and the greedy algorithm: Tight
  worst-case bounds and some generalizations of the rado-edmonds theorem.
\newblock {\em Discret. Appl. Math.}, 7(3):251--274, 1984.

\bibitem[CC14]{ChanC14}
Hau Chan and Jing Chen.
\newblock Truthful multi-unit procurements with budgets.
\newblock In {\em Web and Internet Economics - 10th International Conference,
  {WINE} 2014, Beijing, China, December 14-17, 2014. Proceedings}, pages
  89--105, 2014.

\bibitem[CC16]{ChanC16}
Hau Chan and Jing Chen.
\newblock Budget feasible mechanisms for dealers.
\newblock In {\em Proceedings of the 2016 International Conference on
  Autonomous Agents {\&} Multiagent Systems, Singapore, May 9-13, 2016}, pages
  113--122, 2016.

\bibitem[CGL11]{chen2011approximability}
Ning Chen, Nick Gravin, and Pinyan Lu.
\newblock On the approximability of budget feasible mechanisms.
\newblock In {\em Proceedings of the twenty-second annual ACM-SIAM symposium on
  Discrete Algorithms}, pages 685--699. Society for Industrial and Applied
  Mathematics, 2011.

\bibitem[CRV17]{CRV17}
Vaggos Chatziafratis, Tim Roughgarden, and Jan Vondr{\'{a}}k.
\newblock Stability and recovery for independence systems.
\newblock In {\em 25th Annual European Symposium on Algorithms, {ESA} 2017,
  September 4-6, 2017, Vienna, Austria}, pages 26:1--26:15, 2017.

\bibitem[DH18]{DH18}
Daniel Dadush and Sophie Huiberts.
\newblock A friendly smoothed analysis of the simplex method.
\newblock In Ilias Diakonikolas, David Kempe, and Monika Henzinger, editors,
  {\em Proceedings of the 50th Annual {ACM} {SIGACT} Symposium on Theory of
  Computing, {STOC} 2018, Los Angeles, CA, USA, June 25-29, 2018}, pages
  390--403. {ACM}, 2018.

\bibitem[DK11]{DK11}
Abhimanyu Das and David Kempe.
\newblock Submodular meets spectral: greedy algorithms for subset selection,
  sparse approximation and dictionary selection.
\newblock In {\em Proceedings of the 28th International Conference on
  International Conference on Machine Learning}, pages 1057--1064, 2011.

\bibitem[DPS11]{DobzinskiPS11}
Shahar Dobzinski, Christos~H. Papadimitriou, and Yaron Singer.
\newblock Mechanisms for complement-free procurement.
\newblock In {\em Proceedings 12th {ACM} Conference on Electronic Commerce
  (EC-2011), San Jose, CA, USA, June 5-9, 2011}, pages 273--282, 2011.

\bibitem[EG14]{EG14}
Ludwig Ensthaler and Thomas Giebe.
\newblock A dynamic auction for multi-object procurement under a hard budget
  constraint.
\newblock {\em Research Policy}, 43(1):179--189, 2014.

\bibitem[EKDN18]{EKDN18}
Ethan~R Elenberg, Rajiv Khanna, Alexandros~G Dimakis, and Sahand Negahban.
\newblock Restricted strong convexity implies weak submodularity.
\newblock {\em The Annals of Statistics}, 46(6B):3539--3568, 2018.

\bibitem[EN19]{EN19}
Alina Ene and Huy~L. Nguyen.
\newblock A nearly-linear time algorithm for submodular maximization with a
  knapsack constraint.
\newblock In Christel Baier, Ioannis Chatzigiannakis, Paola Flocchini, and
  Stefano Leonardi, editors, {\em 46th International Colloquium on Automata,
  Languages, and Programming, {ICALP} 2019, July 9-12, 2019, Patras, Greece},
  volume 132 of {\em LIPIcs}, pages 53:1--53:12. Schloss Dagstuhl -
  Leibniz-Zentrum f{\"{u}}r Informatik, 2019.

\bibitem[Fac20]{2020facebook}
Facebook.
\newblock Facebook ad library report.
\newblock \url{https://www.facebook.com/ads/library/report/}, 2020.
\newblock Accessed: 2020-03-26.

\bibitem[Fei98]{Feige98}
Uriel Feige.
\newblock A threshold of ln n for approximating set cover.
\newblock {\em Journal of the ACM (JACM)}, 45(4):634--652, 1998.

\bibitem[FK14]{FeldmanK14}
Vitaly Feldman and Pravesh Kothari.
\newblock Learning coverage functions and private release of marginals.
\newblock In {\em Proceedings of The 27th Conference on Learning Theory, {COLT}
  2014, Barcelona, Spain, June 13-15, 2014}, pages 679--702, 2014.

\bibitem[GJLZ19]{gravin2019optimal}
Nick Gravin, Yaonan Jin, Pinyan Lu, and Chenhao Zhang.
\newblock Optimal budget-feasible mechanisms for additive valuations.
\newblock In {\em Proceedings of the 2019 ACM Conference on Economics and
  Computation}, pages 887--900. ACM, 2019.

\bibitem[GNS14]{GoelNS14}
Gagan Goel, Afshin Nikzad, and Adish Singla.
\newblock Mechanism design for crowdsourcing markets with heterogeneous tasks.
\newblock In {\em Proceedings of the Seconf {AAAI} Conference on Human
  Computation and Crowdsourcing, {HCOMP} 2014, November 2-4, 2014, Pittsburgh,
  Pennsylvania, {USA}}, 2014.

\bibitem[GRS16]{GRS16}
Rishi Gupta, Tim Roughgarden, and C.~Seshadhri.
\newblock Decompositions of triangle-dense graphs.
\newblock {\em {SIAM} J. Comput.}, 45(2):197--215, 2016.

\bibitem[HIM14]{HorelIM14}
Thibaut Horel, Stratis Ioannidis, and S.~Muthukrishnan.
\newblock Budget feasible mechanisms for experimental design.
\newblock In {\em {LATIN} 2014: Theoretical Informatics - 11th Latin American
  Symposium, Montevideo, Uruguay, March 31 - April 4, 2014. Proceedings}, pages
  719--730, 2014.

\bibitem[HS16]{HorelS16}
Thibaut Horel and Yaron Singer.
\newblock Maximization of approximately submodular functions.
\newblock In {\em Advances in Neural Information Processing Systems 29: Annual
  Conference on Neural Information Processing Systems 2016, December 5-10,
  2016, Barcelona, Spain}, pages 3045--3053, 2016.

\bibitem[HS17]{HS17}
Avinatan Hassidim and Yaron Singer.
\newblock Submodular optimization under noise.
\newblock In {\em Proceedings of the 30th Conference on Learning Theory, {COLT}
  2017, Amsterdam, The Netherlands, 7-10 July 2017}, pages 1069--1122, 2017.

\bibitem[KG14]{KG14}
Andreas Krause and Daniel Golovin.
\newblock Submodular function maximization.
\newblock In {\em Tractability: Practical Approaches to Hard Problems}.
  Cambridge University Press, February 2014.

\bibitem[KKT15]{KKT15}
David Kempe, Jon~M. Kleinberg, and {\'{E}}va Tardos.
\newblock Maximizing the spread of influence through a social network.
\newblock {\em Theory of Computing}, 11:105--147, 2015.

\bibitem[KL14]{KL14}
Sanjeev Khanna and Brendan Lucier.
\newblock Influence maximization in undirected networks.
\newblock In {\em Proceedings of the Twenty-Fifth Annual {ACM-SIAM} Symposium
  on Discrete Algorithms, {SODA} 2014, Portland, Oregon, USA, January 5-7,
  2014}, pages 1482--1496, 2014.

\bibitem[KT18]{KhalilabadiT18}
Pooya~Jalaly Khalilabadi and {\'{E}}va Tardos.
\newblock Simple and efficient budget feasible mechanisms for monotone
  submodular valuations.
\newblock In {\em Web and Internet Economics - 14th International Conference,
  {WINE} 2018, Oxford, UK, December 15-17, 2018, Proceedings}, pages 246--263,
  2018.

\bibitem[LMSZ17]{LeonardiMSZ17}
Stefano Leonardi, Gianpiero Monaco, Piotr Sankowski, and Qiang Zhang.
\newblock Budget feasible mechanisms on matroids.
\newblock In {\em Integer Programming and Combinatorial Optimization - 19th
  International Conference, {IPCO} 2017, Waterloo, ON, Canada, June 26-28,
  2017, Proceedings}, pages 368--379, 2017.

\bibitem[LV18]{LV18}
Paul Liu and Jan Vondrak.
\newblock Submodular optimization in the mapreduce model.
\newblock In {\em 2nd Symposium on Simplicity in Algorithms (SOSA 2019)}.
  Schloss Dagstuhl-Leibniz-Zentrum fuer Informatik, 2018.

\bibitem[LZY20]{LiZY20}
Juan Li, Yanmin Zhu, and Jiadi Yu.
\newblock Redundancy-aware and budget-feasible incentive mechanism in crowd
  sensing.
\newblock {\em Comput. J.}, 63(1):66--79, 2020.

\bibitem[NS20]{NS20}
Zeev Nutov and Elad Shoham.
\newblock Practical budgeted submodular maximization.
\newblock {\em CoRR}, abs/2007.04937, 2020.

\bibitem[NSKK16]{NushiS0K16}
Besmira Nushi, Adish Singla, Andreas Krause, and Donald Kossmann.
\newblock Learning and feature selection under budget constraints in
  crowdsourcing.
\newblock In {\em Proceedings of the Fourth {AAAI} Conference on Human
  Computation and Crowdsourcing, {HCOMP} 2016, 30 October - 3 November, 2016,
  Austin, Texas, {USA}}, pages 159--168, 2016.

\bibitem[NW78]{NemhauserW78}
George~L Nemhauser and Laurence~A Wolsey.
\newblock Best algorithms for approximating the maximum of a submodular set
  function.
\newblock {\em Mathematics of operations research}, 3(3):177--188, 1978.

\bibitem[NWF78]{NemhauserWF78}
George~L Nemhauser, Laurence~A Wolsey, and Marshall~L Fisher.
\newblock An analysis of approximations for maximizing submodular set
  functions.
\newblock {\em Mathematical Programming}, 14(1):265--294, 1978.

\bibitem[RDS{\etalchar{+}}15]{ImageNet}
Olga Russakovsky, Jia Deng, Hao Su, Jonathan Krause, Sanjeev Satheesh, Sean Ma,
  Zhiheng Huang, Andrej Karpathy, Aditya Khosla, Michael Bernstein,
  Alexander~C. Berg, and Li~Fei-Fei.
\newblock {ImageNet Large Scale Visual Recognition Challenge}.
\newblock {\em International Journal of Computer Vision (IJCV)},
  115(3):211--252, 2015.

\bibitem[Sin10]{singer2010budget}
Yaron Singer.
\newblock Budget feasible mechanisms.
\newblock In {\em 2010 IEEE 51st Annual Symposium on Foundations of Computer
  Science}, pages 765--774. IEEE, 2010.

\bibitem[SM13]{SingerM13}
Yaron Singer and Manas Mittal.
\newblock Pricing mechanisms for crowdsourcing markets.
\newblock In {\em 22nd International World Wide Web Conference, {WWW} '13, Rio
  de Janeiro, Brazil, May 13-17, 2013}, pages 1157--1166, 2013.

\bibitem[ST04]{SpielmanT04}
Daniel~A. Spielman and Shang{-}Hua Teng.
\newblock Smoothed analysis of algorithms: Why the simplex algorithm usually
  takes polynomial time.
\newblock {\em J. {ACM}}, 51(3):385--463, 2004.

\bibitem[ST19]{ST19}
Grant Schoenebeck and Biaoshuai Tao.
\newblock Influence maximization on undirected graphs: Towards closing the
  (1-1/e) gap.
\newblock In {\em Proceedings of the 2019 {ACM} Conference on Economics and
  Computation, {EC} 2019, Phoenix, AZ, USA, June 24-28, 2019}, pages 423--453,
  2019.

\bibitem[STY20]{STY20}
Grant Schoenebeck, Biaoshuai Tao, and Fang-Yi Yu.
\newblock Limitations of greed: Influence maximization in undirected networks
  re-visited.
\newblock In {\em AAMAS 2020}, 2020.
\newblock To appear.

\bibitem[Svi04]{Sviridenko04}
Maxim Sviridenko.
\newblock A note on maximizing a submodular set function subject to a knapsack
  constraint.
\newblock {\em Oper. Res. Lett.}, 32(1):41--43, 2004.

\bibitem[SVW17]{SviridenkoVW17}
Maxim Sviridenko, Jan Vondr{\'{a}}k, and Justin Ward.
\newblock Optimal approximation for submodular and supermodular optimization
  with bounded curvature.
\newblock {\em Math. Oper. Res.}, 42(4):1197--1218, 2017.

\bibitem[TSP20]{TSP20}
Alfredo Torrico, Mohit Singh, and Sebastian Pokutta.
\newblock On the unreasonable effectiveness of the greedy algorithm: Greedy
  adapts to sharpness.
\newblock {\em CoRR}, abs/2002.04063, 2020.

\bibitem[Von13]{Vondrak13}
Jan Vondr{\'{a}}k.
\newblock Symmetry and approximability of submodular maximization problems.
\newblock {\em {SIAM} J. Comput.}, 42(1):265--304, 2013.

\bibitem[WS98]{WS98}
Duncan Watts and Steven Strogatz.
\newblock Collective dynamics of 'small-world' networks.
\newblock {\em Nature}, 1998.

\bibitem[Yos16]{Yoshida16}
Yuichi Yoshida.
\newblock Maximizing a monotone submodular function with a bounded curvature
  under a knapsack constraint.
\newblock {\em CoRR}, abs/1607.04527, 2016.

\bibitem[ZLM16]{ZhaoLM16}
Dong Zhao, Xiang{-}Yang Li, and Huadong Ma.
\newblock Budget-feasible online incentive mechanisms for crowdsourcing tasks
  truthfully.
\newblock {\em {IEEE/ACM} Trans. Netw.}, 24(2):647--661, 2016.

\bibitem[ZWG{\etalchar{+}}17]{ZhengWGZTC17}
Zhenzhe Zheng, Fan Wu, Xiaofeng Gao, Hongzi Zhu, Shaojie Tang, and Guihai Chen.
\newblock A budget feasible incentive mechanism for weighted coverage
  maximization in mobile crowdsensing.
\newblock {\em {IEEE} Trans. Mob. Comput.}, 16(9):2392--2407, 2017.

\end{thebibliography}

\appendix
\section{Explicit Analysis for Two Budgets}\label{sec:2-budgets}
In this section, we establish an analytic formula for the optimal expected approximation ratio for the uniform distribution of two budgets $k_1=\rho\cdot k_2<k_2$ for every sufficiently large $k_1$ and $k_2-k_1$ that grow with the size of instance. This is done by implementing the analysis of Theorem~\ref{thm:many_budgets} explicitly. We start by re-stating the hard instances.
\paragraph{Construction of hard instances.} Using Theorem~\ref{thm:basic_hardness}, we create two hard (with respect to an arbitrary efficient algorithm) submodular functions $f_{k_1}$ and $f_{k_2-k_1}$ with disjoint ground sets $V_1$ and $V_2$. To simplify notation, we let $f_1$ and $f_2$ denote $f_{k_1}$ and $f_{k_2-k_1}$ respectively. Then, we normalize the two functions such that $f_1(O^{(1)})=f_2(O^{(2)})=1$, where $O^{(1)}$ denotes the optimal size-$k_1$ solution for $f_{1}$ and $O^{(2)}$ is the optimal size-$(k_2-k_1)$ solution for $f_{2}$. Furthermore, we extend both the functions to the ground set $V:=V_1\cup V_2$ in a natural way that the sets of $V_2$ have zero value to $f_1$ and vice versa. Finally, for $\alpha>0$, we define $f:V\to\R_{\ge 0}$ as $f(X)=\alpha\cdot f_1(X)+f_2(X)$. \\

\begin{lemma}\label{lem:gap_submod}
For any efficient algorithm $\mathcal{A}$, there is a submodular function $f$ constructed as above that has the following properties:
\begin{enumerate}[label=(\roman*)]
    \item For every $k=\beta\cdot k_2$ with $0\le\beta\le 1$, the optimal value of $f$ with budget $k$ is $\min\{1,\beta/(1-\rho)\}+\max\{(\beta-1+\rho)/\rho,0\}\cdot\alpha$ if $\alpha<\rho/(1-\rho)$ and is $\min\{1,\beta/\rho\}\cdot\alpha+\max\{(\beta-\rho)/(1-\rho),0\}$ otherwise.
    \item For every $k=\beta\cdot k_2$ with $0\le\beta\le 1$, the solution value of $\mathcal{A}$ on $f$ with budget $k$ is upper bounded by $\max\limits_{0\le x\le\beta} (1-e^{-x/\rho})\cdot \alpha+(1-e^{(x-\beta)/(1-\rho)})$.
\end{enumerate}
\end{lemma}
\begin{proof}
\textit{(i)} By the first property in Theorem~\ref{thm:basic_hardness}, the optimal values of $f_1,f_2$ grow linearly with the budget. Moreover, the marginal gain of an element of $O^{(1)}$ is $\alpha/k_1=\alpha/(\rho\cdot k_2)$ until we select all $O^{(1)}$ and is zero after that. The marginal gain of an element in $O^{(2)}$ is $1/(k_2-k_1)=1/((1-\rho)k_2)$ until $O^{(2)}$ is exhausted. Hence, if $\alpha< \rho/(1-\rho)$, the optimal solution to $f$ should prefer the elements of $O^{(2)}$ until it exhausts $O^{(2)}$, and then spend the rest of budget on $O^{(1)}$. Hence, the optimal value is $\min\{1,\beta/(1-\rho)\}+\max\{(\beta-1+\rho)/\rho,0\}\cdot\alpha$. The other case is similar.

\textit{(ii)} Suppose $x\cdot k_2=(x/\rho) k_1$ elements are chosen from $V_1$, then, the remaining $((\beta-x)/(1-\rho))k_2$ elements are from $V_2$. By the second property in Theorem~\ref{thm:basic_hardness}, the value we can obtain is at most $(1-e^{-x/\rho})\cdot \alpha+(1-e^{(x-\beta)/(1-\rho)})$. Therefore, the maximum of this objective is an upper bound of the optimum of $f$ with budget $k$.
\end{proof}

With this lemma, we can easily prove a parametrized hardness result.

\begin{proposition}\label{prop:hardness_submod}
Given any $0<\rho<1$ and $\alpha(1-\rho)/\rho\ge 1$, for any $\epsilon>0$, there is no efficient algorithm can approximate submodular maximization problem for two budgets $k_2$ and $k_1=\rho\cdot k_2$, with the average approximation ratio is larger than
\begin{enumerate}[label=(\roman*)]
    \item 
        \begin{equation}\label{eq:hardness_ratio_1}
        \begin{split}
            \frac{1}{2}&\Bigg(\frac{\left(1-e^{-\rho}\left(\frac{\rho}{\alpha(1-\rho)}\right)^{1-\rho}\right)\alpha+1-e^{-\rho}\left(\frac{\alpha(1-\rho)}{\rho}\right)^\rho}{\alpha}\\
            &+\frac{\left(1-e^{-1}\left(\frac{\rho}{\alpha(1-\rho)}\right)^{1-\rho}\right)\alpha+1-e^{-1}\left(\frac{\alpha(1-\rho)}{\rho}\right)^\rho}{\alpha+1}\Bigg)+\epsilon,
        \end{split}
        \end{equation}
        if $\alpha(1-\rho)/\rho\le e$,
    \item
        \begin{equation}\label{eq:hardness_ratio_2}
        \begin{split}
            \frac{1}{2}\Bigg((1-e^{-1})
            +\frac{\left(1-e^{-1}\left(\frac{\rho}{\alpha(1-\rho)}\right)^{1-\rho}\right)\alpha+1-e^{-1}\left(\frac{\alpha(1-\rho)}{\rho}\right)^\rho}{\alpha+1}\Bigg)+\epsilon,
        \end{split}
        \end{equation}
        if $e\le\alpha(1-\rho)/\rho\le e^{1/\rho}$,
    \item
        \begin{equation}\label{eq:hardness_ratio_3}
        \begin{split}
            \frac{1}{2}\bigg((1-e^{-1})
            +(1-e^{-1/\rho})\cdot\frac{\alpha}{\alpha+1}\bigg)+\epsilon,
        \end{split}
        \end{equation}
        if $\alpha(1-\rho)/\rho\ge e^{1/\rho}$.
\end{enumerate}
\end{proposition}
\begin{proof}
Since $\alpha(1-\rho)/\rho\ge 1$, by the first property in Lemma~\ref{lem:gap_submod}, we know the optimal value for budget $k_1$ is $\alpha$ and that for budget $k_2$ is $\alpha+1$. We can maximize the best achievable solution value $(1-e^{-x/\rho})\cdot \alpha+(1-e^{(x-\beta)/(1-\rho)})$ for $\beta=\rho$ and $1$ by standard calculus. In general, we find that the optimal $x$ is $(\ln(\alpha(1-\rho)/\rho)+\beta/(1-\rho))/(1/(1-\rho)+1/\rho)$. Then, we observe that if $1\le \alpha(1-\rho)/\rho\le e$, the optimal $x$ for $\beta=\rho$ is between $0$ and $\rho$, and that for $\beta=1$ is between $\rho$ and $1$. Hence the optimal $x$ for both $\beta$ are feasible, and we can calculate the analytic formula of each maximum. Therefore, we have an upper bound of approximation ratio for each budget. Obviously, the average of these two upper bounds, which is given in Eq.~\eqref{eq:hardness_ratio_1}, is an upper bound for the average approximation ratio. If $e<\alpha(1-\rho)/\rho\le e^{1/\rho}$, then the optimal $x$ is $\rho$ when $\beta=\rho$ and is still $(\ln(\alpha(1-\rho)/\rho)+\beta/(1-\rho))/(1/(1-\rho)+1/\rho)$ when $\beta=1$. As before, we can calculate the upper bound, which is given in Eq.~\eqref{eq:hardness_ratio_2}. Finally, if $\alpha(1-\rho)/\rho\ge e^{1/\rho}$, then the optimal $x$ is $\rho$ when $\beta=\rho$ and is $1$ when $\beta=1$. The corresponding upper bound is given in Eq.~\eqref{eq:hardness_ratio_3}.
\end{proof}
Next, we derive the closed-form parametrized formulas of the approximation ratios of the greedy algorithm for monotone submodular maximization with two budgets, which will match the hardness in Proposition~\ref{prop:hardness_submod}. Before that, we establish a useful lemma for greedy analysis.
\begin{lemma}\label{lem:greedy}
Given the same conditions as in Lemma~\ref{lem:greedy_one_step}, for all $k_1=\theta\cdot k_2$ with $\theta\ge 0$ and $k=\eta\cdot k_1$ with $0\le\eta\le1$, the following inequality holds,
\[
f(X_{k_1})\ge (1-e^{\theta\cdot\eta-\theta})f(O_{k_2})+e^{\theta\cdot\eta-\theta}\cdot f(X_{k}),
\]
and in particular, by letting $\eta=0$,
\[
f(X_{k_1})\ge (1-e^{-\theta})f(O_{k_2}).
\]
\end{lemma}
\begin{proof}
We start from Lemma~\ref{lem:greedy_one_step},
\[ 
    f(X_i)-f(X_{i-1})\ge
    \frac{1}{k_2}(f(O_{k_2})-f(X_{i-1})).
\]
We rearrange the terms as follows,
\[
    f(O_{k_2})-f(X_{i}) \le \left(1-\frac{1}{k_2}\right)\cdot(f(O_{k_2})-f(X_{i-1})),
\]
and we recursively apply this step and get
\begin{align*}
    f(O_{k_2})-f(X_{k_1}) &\le \left(1-\frac{1}{k_2}\right)^{k_1-k}\cdot(f(O_{k_2})-f(X_{k})) \\
    &\le e^{\theta-\theta\cdot\eta}(f(O_{k_2})-f(X_{k})).
\end{align*}
The proof finishes by rearranging the terms.
\end{proof}
\begin{proposition}\label{prop:two_budget}
Given a monotone submodular function $f$ and two budgets $k_2$ and $k_1=\rho\cdot k_2$ with $0<\rho<1$, we we let $X_k$ and $O_k$ denote the greedy solution and the optimal solution of cardinality $k$. Suppose $f(O_{k_1})=c\cdot f(O_{k_2})$, where $\rho\le c\le1$. Then, the greedy algorithm has the following average approximation ratios,
\begin{enumerate}[label=(\roman*)]
    \item 
        \begin{equation}\label{eq:greedy_ratio_1}
        \begin{split}
            \frac{1}{2}&\bigg(\bigg(\left(1-e^{-\rho}\left(\frac{1-\rho}{\rho(1/c-1)}\right)^\rho\right)\cdot\frac{1}{c} + e^{-\rho}\left(\frac{1-\rho}{\rho(1/c-1)}\right)^\rho\cdot\frac{1-\rho/c}{1-\rho}\bigg)\\
            &+\bigg(1-e^{-1}\left(\frac{1-\rho}{\rho(1/c-1)}\right)^\rho+e^{-1}\left(\frac{1-\rho}{\rho(1/c-1)}\right)^\rho\cdot\frac{c-\rho}{1-\rho}\bigg)\bigg),
        \end{split}    
        \end{equation}
        if $((1-\rho/c)/(1-\rho))\le 1-e^{-1}$,
    \item
        \begin{equation}\label{eq:greedy_ratio_2}
        \begin{split}
            \frac{1}{2}&\Bigg((1-e^{-1})+\left(1-e^{-1}\left(\frac{1-\rho}{\rho(1/c-1)}\right)^\rho\cdot(1-c)+e^{-1}\left(\frac{1-\rho}{\rho(1/c-1)}\right)^{1-\rho}c\right)\Bigg),
        \end{split}    
        \end{equation}
        if $1-e^{-1}\le ((1-\rho/c)/(1-\rho))\le 1/c$,
    \item
        \begin{equation}\label{eq:greedy_ratio_3}
        \begin{split}
            \frac{1}{2}&(((1-e^{-1})+(1-e^{-1/\rho})c)),
        \end{split}    
        \end{equation}
        if $((1-\rho/c)/(1-\rho))\ge 1/c$.
\end{enumerate}
Moreover, there is no efficient algorithm can achieve better approximation ratios.
\end{proposition}
\begin{proof}
By Lemma~\ref{lem:greedy_one_step}, we have the following two guarantees,
\begin{align*}
    &f(X_i)-f(X_{i-1})\ge
    \frac{1}{k_2}(f(O_{k_2})-f(X_{i-1})),\\
    &f(X_i)-f(X_{i-1})\ge
    \frac{1}{k_1}(f(O_{k_1})-f(X_{i-1})).
\end{align*}
Observe that $(f(O_{k_1})-f(X_{i-1}))/k_1\ge (f(O_{k_2})-f(X_{i-1}))/k_2$ holds if and only if $f(X_{i-1})\le((c-\rho)/(1-\rho))f(O_{k_2})$, where the right hand side is equal to $((1-\rho/c)/(1-\rho))f(O_{k_1})$. We let $i^*-1=\tau\cdot k_1$ be the largest $i-1$ such that $f(X_{i-1})\le((1-\rho/c)/(1-\rho))f(O_{k_1})$ holds. We first consider the case where
\begin{equation}\label{eq:greedy_case_1}
    ((1-\rho/c)/(1-\rho))\le 1-e^{-1},
\end{equation}
which implies that $\tau\le1$, because $f(X_{k_1})\ge(1-e^{-1})f(O_{k_1})$. Then, by Lemma~\ref{lem:greedy} with $\theta=\tau,\eta=0$, $f(X_{i^*-1})\ge (1-e^{-\tau})f(O_{k_1})$. It follows that
\begin{equation}\label{eq:bound_of_tau}
    ((1-\rho/c)/(1-\rho))\ge 1-e^{-\tau}.
\end{equation}
Now we apply Lemma~\ref{lem:greedy} with $\theta=\rho,\eta=i^*/k_1=\tau+o(1)$,
\begin{align*}
    f(X_{k_1})&\ge (1-e^{\rho\cdot\tau-\rho})f(O_{k_2})+e^{\rho\cdot\tau-\rho}\cdot f(X_{i^*}) \\
    &\ge (1-e^{\rho\cdot\tau-\rho})f(O_{k_2})+e^{\rho\cdot\tau-\rho} ((1-\rho/c)/(1-\rho))f(O_{k_1}) \\
    &= ((1-e^{\rho\cdot\tau-\rho})/c+e^{\rho\cdot\tau-\rho}((1-\rho/c)/(1-\rho)))f(O_{k_1}) \\
    &\ge \bigg(\left(1-e^{-\rho}\left(\frac{1-\rho}{\rho(1/c-1)}\right)^\rho\right)\cdot\frac{1}{c} + e^{-\rho}\left(\frac{1-\rho}{\rho(1/c-1)}\right)^\rho\cdot\frac{1-\rho/c}{1-\rho}\bigg)f(O_{k_1}),
\end{align*}
where the second inequality is by definition of $i^*$, and the last inequality follows from Eq.~\eqref{eq:bound_of_tau}. Then, we apply Lemma~\ref{lem:greedy} with $\theta=1,\eta=\rho$ and use the previous bound of $f(X_{k_1})$,
\begin{align*}
    f(X_{k_2})&\ge
    (1-e^{\rho-1})f(O_{k_2})+e^{\rho-1}f(X_{k_1})\\
    &\ge \bigg(1-e^{-1}\left(\frac{1-\rho}{\rho(1/c-1)}\right)^\rho+e^{-1}\left(\frac{1-\rho}{\rho(1/c-1)}\right)^\rho\cdot\frac{c-\rho}{1-\rho}\bigg)f(O_{k_2}).
\end{align*}
We relate this case, where we assume Eq.~\eqref{eq:greedy_case_1}, with the first case of Proposition~\ref{prop:hardness_submod} by noticing that $c=\alpha/(1+\alpha)$. It is straightforward to verify the approximation ratios for $f(X_{k_1})$ and $f(X_{k_2})$ match the ratios there. Therefore, greedy is optimal for this case. Next, we consider the case where
\begin{equation}\label{eq:greedy_case_2}
    1-e^{-1}\le ((1-\rho/c)/(1-\rho))\le 1/c,
\end{equation}
which implies $1\le\tau\le1/\rho$, and Eq.~\eqref{eq:bound_of_tau} still holds. In this case, we know $f(X_{k_1})\ge(1-e^{-1})f(O_{k_1})$, and similar to before, we apply Lemma~\ref{lem:greedy} with $\theta=1,\eta=i^*/k_2=\rho\cdot\tau+o(1)$,
\begin{align*}
    f(X_{k_2})&\ge (1-e^{\rho\cdot\tau-1})f(O_{k_2})+e^{\rho\cdot\tau-1}\cdot f(X_{i^*}) \\
    &\ge (1-e^{\rho\cdot\tau-1})f(O_{k_2})+e^{\rho\cdot\tau-1}(1-e^{-\tau})f(O_{k_1}) \\
    &= ((1-e^{\rho\cdot\tau-1})+e^{\rho\cdot\tau-1}(1-e^{-\tau})c)f(O_{k_2}) \\
    &\ge \left(1-e^{-1}\left(\frac{1-\rho}{\rho(1/c-1)}\right)^\rho\cdot(1-c)+e^{-1}\left(\frac{1-\rho}{\rho(1/c-1)}\right)^{1-\rho}c\right)f(O_{k_2}),
\end{align*}
where the second inequality is again by Lemma~\ref{lem:greedy} with $\theta=1,\eta=i^*/k_1=\tau+o(1)$ and the last inequality follows from Eq.~\eqref{eq:bound_of_tau}. It is not hard to verify that this case correspond to the second case of Proposition~\ref{prop:hardness_submod} and greedy has optimal ratios. Finally, we consider the last case where
\begin{equation}\label{eq:greedy_case_3}
    ((1-\rho/c)/(1-\rho))\ge 1/c.
\end{equation}
In this case, we know that $f(X_{k_1})\ge(1-e^{-1})f(O_{k_1})$ and $f(X_{k_2})\ge(1-e^{-1/\rho})f(O_{k_1})=(1-e^{-1/\rho})c\cdot f(O_{k_2})$. We conclude that greedy is optimal by comparing this case with the third case of Proposition~\ref{prop:hardness_submod}.
\end{proof}

\begin{figure}[h]
\centering
    \includegraphics[scale=0.7]{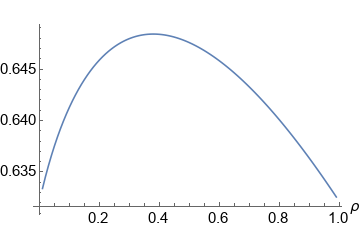}
    \caption{Optimal average approximation ratios for two budgets.}
    \label{fig:two_budget}
\end{figure}

For every $0<\rho<1$, using Proposition~\ref{prop:two_budget}, we can compute the worst $c$ to minimize the approximation ratio, and it follows that the minimal ratio is the best achievable approximation guarantee for submodular maximization with budgets $k_2$ and $k_1=\rho\cdot k_2$. It turns out that the first case of Proposition~\ref{prop:two_budget} is always the worst case. We illustrate the best achievable approximation ratios for $0.01\le\rho\le0.99$ in Figure~\ref{fig:two_budget}.

\section{Optimal Algorithms in Practical Settings}\label{sec:practical_algorithms}
In this section, we show that our main result, greedy is optimal for multiple budgets, generalizes to the constant rounds Map-Reduce algorithm in distributed setting~\cite{LV18}, and the logarithmic rounds parallel algorithm~\cite{BRS19a}. We sketch the main ideas behind these algorithms and point out how to adapt them to Observation~\ref{obs:opt_algorithms}. Before that, we provide the proof of Observation~\ref{obs:opt_algorithms}.

\begin{observation}[Observation~\ref{obs:opt_algorithms} restated]
For any perturbation factors $0<\rho_1<\rho_2<\dots<\rho_{m}$, there exists a sufficiently large $k$ that grows with the size of instance such that given $m$ budgets $k_1=\rho_1\cdot k$, \dots, $k_{m}=\rho_{m}\cdot k$, the optimality described in Theorem~\ref{thm:many_budgets} actually holds for a general class of algorithms such that:
\begin{itemize}
    \item Given budget $k_i$, the algorithm $\mathcal{A}$ runs in $T$ rounds ($T$ is sufficiently large), of which each round selects about $k_i/T$ elements.
    \item For any $\epsilon>0$, it holds for all $t\in[T]$, for all $j\in[m]$, that $f(X^{\mathcal{A}}_{tk_i/T})-f(X^{\mathcal{A}}_{(t-1)k_i/T})\ge ((1-\epsilon)\rho_i/(\rho_j T))\cdot(f(O_{k_j})-f(X^{\mathcal{A}}_{tk_i/T}))$, where $X^{\mathcal{A}}_s$ is the $s$-th element chosen by $\mathcal{A}$.
\end{itemize}
\end{observation}
\begin{proof}
We let $\hat{f}(X^{\textrm{greedy}}_{l})$ denote the lower bound estimate of $f(X^{\textrm{greedy}}_{l})$ that we get by iteratively applying the best greedy guarantees until the $l$-th iteration ($\hat{f}(X^{\mathcal{A}}_{l})$ is defined similarly). Note that the second property in the observation is similar to the performance guarantee of greedy algorithm with respect to each $O_{k_j}$. The only difference is that with respect to any $O_{k_j}$, in average, every element selected in round $t$ of $\mathcal{A}$ has the same guarantee $(f(O_{k_j})-\hat{f}(X^{\mathcal{A}}_{tk_i/T}))/k_j$ (ignore the $1-\epsilon$ factor), while for each $s\le k_{i}/T$, the $((t-1)k_{i}/T+s)$-th element selected by standard greedy has guarantee $(f(O_{k_j})-\hat{f}(X^{\textrm{greedy}}_{(t-1)k_{i}/T+s-1}))/k_j$. However, we can show that this difference between the two guarantees can be ignored. First, observe that $\hat{f}(X^{\mathcal{A}}_{tk_{i}/T})\le \hat{f}(X^{\textrm{greedy}}_{tk_{i}/T})$ for all $t\le T$, because greedy has better choices of guarantees than $\mathcal{A}$. Moreover, notice that for all $t$, $(f(O_{k_j})-\hat{f}(X^{\textrm{greedy}}_{(t-1)k_{i}/T+s-1}))/k_j\le (f(O_{k_j})-\hat{f}(X^{\textrm{greedy}}_{(t-1)k_{i}/T}))/k_j$, and the difference between $(f(O_{k_j})-\hat{f}(X^{\textrm{greedy}}_{(t-1)k_{i}/T}))/k_j$ and $(f(O_{k_j})-\hat{f}(X^{\mathcal{A}}_{tk_{i}/T}))/k_j$ for any $j$ is upper bounded by $(\hat{f}(X^{\mathcal{A}}_{tk_i/T})-\hat{f}(X^{\textrm{greedy}}_{(t-1)k_{i}/T}))/k_1$, which in turn is upper bounded by $\epsilon_t:=(\hat{f}(X^{\mathcal{A}}_{tk_i/T})-\hat{f}(X^{\mathcal{A}}_{(t-1)k_{i}/T}))/k_1$. Furthermore, if we iteratively apply the guarantee for each $t\le T$ and $s\le k_i/T$ for $\mathcal{A}$ as follows
\begin{align*}
    \hat{f}(X^{\mathcal{A}}_{(t-1)k_{i}/T+s})&\ge \hat{f}(X^{\mathcal{A}}_{(t-1)k_{i}/T+s-1})-\hat{f}(X^{\mathcal{A}}_{tk_{i}/T})/k_j+f(O_{k_j})/k_j \\
    &= \hat{f}(X^{\mathcal{A}}_{(t-1)k_{i}/T+s-1})-\hat{f}(X^{\textrm{greedy}}_{(t-1)k_{i}/T+s-1})/k_j+f(O_{k_j})/k_j-\epsilon_t,
\end{align*}
where for each $t$ and $s$, $j$ is chosen to be same as the best choice of $j$ for greedy in this iteration, then by an inductive argument (base case is $\hat{f}(X^{\mathcal{A}}_0)=\hat{f}(X^{\textrm{greedy}}_0)$), we have that $\hat{f}(X^{\textrm{greedy}}_{(t-1)k_{i}/T+s})-\hat{f}(X^{\mathcal{A}}_{(t-1)k_{i}/T+s})\le\sum_{r=1}^{t-1} (k_i/T)\epsilon_r+s\cdot\epsilon_t$ for each $t\le T$ and $s\le k_i/T$. By a telescoping sum, $\sum_{t=1}^T(k_i/T)\epsilon_t\le (f(X^{\mathcal{A}}_{k_i})-f(\emptyset))/(\rho_1 T/\rho_i)$, which is negligible if $T$ is sufficiently large. Therefore, the final performance guarantee of $\mathcal{A}$ is approximately equal to the final greedy guarantee.
\end{proof}

\paragraph{Map-Reduce algorithm.}
Suppose the budget is $k_i$. The setup is that there are $\sqrt{n/k_i}$ machines and a central machine, each with memory $\widetilde{O}(\sqrt{nk_i})$. The algorithm has $t=O(\frac{1}{\epsilon})$ Map-Reduce rounds and maintains a solution set $G$, which is empty initially. At the $l$-th round, the algorithm sets a threshold $\frac{1}{k_i}(1-\frac{1}{t})^l f(O_{k_i})$ and wants to add $\frac{k_i}{t}$ elements to $G$ (actually, it might differ from this amount, but this is fine as we will explain later in this paragraph), each with marginal gain above the threshold. To achieve this, each machine from its storage selects a candidate set consisting of the elements that have marginal gains above the threshold with respect to $G$ ($G$ is not updated) and sends the candidates to the central machine, and then the central machine enumerates all the candidates and adds the element to $G$ if it has marginal gain above the threshold with respect to the latest $G$. The chosen threshold is actually the greedy guarantee of marginal gain when the cumulative utility reaches $(1-(1-\frac{1}{t})^l)f(O_{k_i})$. Hence, during the enumeration procedure on the central machine, either it successfully selects $\frac{k_i}{t}$ elements with marginal contribution above $\frac{1}{k_i}(1-\frac{1}{t})^l f(O_{k_i})$, or there is no such element left, in which case the cumulative utility should already reach $(1-(1-\frac{1}{t})^l)f(O_{k_i})$. In either case, we will achieve roughly $(1-(1-\frac{1}{t})^l)f(O_{k_i})$ at the end of $l$-th round, and the final $1-1/e$ approximation ratio follows by standard greedy analysis. Two issues remain---first, we do not know $f(O_{k_i})$, which can be fixed by standard "guessing optimal value" trick, second, we need to bound the memory usage. For the ordinary machines, we can simply randomly partition the ground set, and for the central machine, this can be fixed as follows: at the beginning of each round, the central machine samples a random set $S$ of size $4\sqrt{n{k_i}}$ and sequentially adds the elements from $S$ to $G$ if the element has marginal gain above the threshold with respect to the latest $G$; If this procedure ends up selecting at least ${k_i}$ elements, the algorithm can stop, otherwise it continues as before. Using a martingale argument, it can be shown that if there are many (more than $\sqrt{n{k_i}}$) candidate elements chosen by the ordinary machines, then with high probability, the central machine should have already chosen at least ${k_i}$ elements in the above procedure.

In order to apply Observation~\ref{obs:opt_algorithms}, we need the greedy guarantees with respect to all $O_{k_j}$'s. To this end, we can guess $f(O_{k_j})$'s rather than just $f(O_{k_i})$, and moreover, we set the threshold at the $l$-th round as the largest of $\frac{1}{k_j}(f(O_{k_j})-\hat{f}(X_{k_il/t}))$ for all $j$, where $\hat{f}(X_{k_il/t})$ is the lower bound estimate of $f(X_{k_il/t})$ we get by applying best guarantee for each iteration of greedy. Finally, if we want the algorithm to be oblivious to the budget distribution, we can simply discretize the domain of perturbed budgets and apply above-mentioned trick for the budgets in the discretized domain.

\paragraph{Parallel algorithm.}
Suppose the budget is $k_i$. The parallel algorithm is similar to the MapReduce algorithm. It runs in $t$ rounds and maintains a solution set $G$. In each round, it adds to $G$ a set of $\frac{k_i}{t}$ elements with total marginal gain above the threshold $\frac{1-\epsilon}{t}(f(O_{k_i})-f(G))$. Specifically, the algorithm first selects a candidate set $X$ by iteratively discarding from $X$ all the elements that have marginal contribution roughly below $\frac{1-\epsilon}{k_i}(f(O_{k_i})-f(G))$ with respect to the union between $G$ and a random subset of size $\frac{k}{t}$ of $X$ until the expected total marginal gain of a random size-$\frac{k_i}{t}$ subset of $X$ achieves the threshold, and then it samples a set of size $\frac{k_i}{t}$ from $X$ and adds it to $G$. Each round terminates quickly because if the expected total marginal gain of a random subset is low in one iteration, then there should be many elements with low marginal contribution, and they will be discarded together in this iteration.

In order to apply our analysis, we can adapt the algorithm similarly to what we did for the MapReduce algorithm, \ie, we set the threshold as the largest of $\frac{1-\epsilon}{k_j}(f(O_{k_j})-f(G))$ for all $j$.

\end{document}